\documentclass[reqno,12pt]{amsart}%
\usepackage{amsfonts}
\usepackage{amsmath}
\usepackage{amssymb}
\usepackage{graphicx}%
\providecommand{\U}[1]{\protect\rule{.1in}{.1in}}
\newtheorem{theorem}{Theorem}[section]

\newtheorem{corollary}[theorem]{Corollary}

\newtheorem{definition}[theorem]{Definition}
\newtheorem{example}[theorem]{Example}

\newtheorem{lemma}[theorem]{Lemma}

\newtheorem{proposition}[theorem]{Proposition}
\newtheorem{remark}[theorem]{Remark}

\numberwithin{theorem}{section}

\begin{document}
\title[A duality theory]{A duality theory for unbounded Hermitian operators in Hilbert space}
\author{Palle E.T. Jorgensen}
\address{Department of Mathematics\\
University of Iowa\\
Iowa City, IA 52242-1466}
\email{jorgen@math.uiowa.edu}
\urladdr{http://www.math.uiowa.edu}
\dedicatory{\emph{University of Iowa}}\date{}
\subjclass[2000]{47B25, 47B32, 47B37, 47S50, 60H25, 81P15, 81Q10}
\keywords{operators in Hilbert space, deficiency spaces, Schr\"{o}dinger equation,
selfadjoint extensions, reproducing kernel Hilbert spaces}

\begin{abstract}
We develop a duality theory for unbounded Hermitian operators with dense
domain in Hilbert space. As is known, the obstruction for a Hermitian operator
to be selfadjoint or to have selfadjoint extensions is measured by a pair of
deficiency indices, and associated deficiency spaces; but in practical
problems, the direct computation of these indices can be difficult. Instead,
in this paper we identify additional structures that throw light on the
problem. While duality considerations are a tested tool in mathematics, we
will attack the problem of computing deficiency spaces for a single Hermitian
operator with dense domain in a Hilbert space which occurs in a duality
relation with a second Hermitian operator, often in the same Hilbert space.

\end{abstract}
\maketitle

\section{Introduction\label{Intro}}

The theory of unbounded Hermitian operators with dense domain in Hilbert space
was developed by H. M. Stone and John von Neumann with view to use in quantum
theory; more precisely to put the spectral theory of the Schr\"{o}dinger
equation on a sound mathematical foundation. Early in the theory, it was
realized that a Hermitian operator may not be selfadjoint. It was given a
quantitative formulation in the form of deficiency indices and deficiency
spaces, and we refer the reader to the books \cite{St90} and \cite{vN32} , and
more recently \cite{DuSc88} and \cite{ReSi75}. In physical problems, see e.g.,
\cite{Bal00}, these mathematical notions of defect take the form of
\textquotedblleft boundary conditions;\textquotedblright\ for example waves
that are diffracted on the boundary of a region in Euclidean space; the
scattering of classical waves on a bounded obstacle \cite{LaPh80}; a quantum
mechanical \textquotedblleft particle\textquotedblright\ in a repulsive
potential that shoots to infinity in finite time; or in more recent
applications (see e.g., \cite{JoPe08}, \cite{BBC+09}, \cite{BrWo05},
\cite{OrWo07}) random walk on infinite weighted graphs $G$ that
\textquotedblleft wonder off\textquotedblright\ to points on an idealized
boundary of $G$. In all of the instances, one is faced with a dynamical
problem: For example, the solution to a Schr\"{o}dinger equation, represents
the time evolution of quantum states in a particular problem in atomic
physics. The operators in these applications will be Hermitian, but in order
to solve the dynamical problems, one must first identify a selfadjoint
extension of the initially given operator. Once that is done, von Neumann's
spectral theorem can then be applied to the selfadjoint operator. A choice of
selfadjoint extension will have a spectral resolution, i.e., it is an integral
of an orthogonal projection valued measure; with the different extensions
representing different \textquotedblleft physical\textquotedblright\ boundary
conditions. Since non-zero deficiency indices measure that degree of
non-selfadjointness, the question of finding selfadjoint extensions takes on
some urgency.

Now the variety of applied problems that lend themselves to computation of
deficiency indices and the study of selfadjoint extensions are vast and
diverse. As a result, it helps if one can identify additional structures that
throw light on the problem. Here duality considerations within the framework
of Hilbert space are tested tools in applied mathematics. In this paper we
will device such a geometric duality theory: We will attack the problem of
computing deficiency spaces for a single Hermitian operator with dense domain
in a Hilbert space which occurs in a duality relation with a second Hermitian
operator, often in the same Hilbert space. We will further use our duality to
prove essential selfadjointness of families of Hermitian operators that arise
naturally in reproducing kernel Hilbert spaces. The latter include graph
Laplacians for infinite weighted graphs $(G,w)$ with the Laplacian in this
context presented as a Hermitian operator in an associated Hilbert space of
finite energy functions on the vertex set in $G$. Other examples include
Hilbert spaces of band-limited signals. Further applications enter into the
techniques used in discrete simulations of stochastic integrals, see
\cite{Hi80}. We encountered the present operator theoretic duality in our
study of discrete Laplacians, which in turn have part of its motivation in
numerical analysis. A key tool in applying numerical analysis to solving
partial differential equations is discretization, and use of repeated
differences; see e.g., \cite{AH05}.

Specifically, one picks a grid size $h$, and then proceeds in steps: (1)
Starting with a partial differential operator, then study an associated
discretized operator with the use of repeated differences on the $h$-lattice
in $\mathbb{R}^{d}$. (2) Solve the discretized problem for $h$ fixed. (3) As
$h$ tends to zero, numerical analysts evaluate the resulting approximation
limits, and they bound the error terms.

Our present approach, based on reproducing kernels and unbounded operators,
fits into a larger framework in applied operator theory, for example the use
of reproducing kernel Hilbert spaces in the determination of optimal spectral
estimation: Here the problem is to estimate some sampled signal represented as
the sum of a deterministic (time-) function and a term representing noise, for
example white noise; see e.g., \cite{Doo59, Hi80}. For the multivariable case,
the process under study is indexed by some prescribed discrete set $X$
(representing sample points; it could be the vertex set in an infinite graph).
The choice of statistical distribution, modeling the noise term, then amounts
to a selection of a reproducing kernel (representing function differences)
with vectors $v_{x}$ (dipoles in the present context), and linear combinations
of these vectors $v_{x}$ in this approach then represents a spectral
estimator. The problem becomes that of selecting samples which minimize error
terms in a prediction of a signal.

\section{Reproducing kernel-Hilbert spaces\label{ReproKHspaces}}

For this purpose, one must use a metric, and the norm in Hilbert space has
proved an effective tool, hence the Hilbert spaces and the operator theory.
This procedure connects to our present graph-Laplacians: When discretization
is applied to the Laplace operator in $d$ continuous variables, the result is
the graph of integer points $\mathbb{Z}^{d}$ with constant weights. But if
numerical analysis is applied instead to a continuous Laplace operator on a
Riemannian manifold, the discretized Laplace operator will instead involve
infinite graph with variable weights, so with vertices in other configurations
than $\mathbb{Z}^{d}$. Inside the technical sections we will use standard
tools from analysis and probability. References to the fundamentals include
\cite{Doo59}, \cite{Kol77}, \cite{Nel73} and \cite{YN08}. There is a large
literature covering the general theory of reproducing kernel Hilbert spaces
and its applications, see e.g., \cite{Alp06}, \cite{AL08a}, \cite{AL08b}%
,\cite{Aro50}, and \cite{Zh09}. Such applications include potential theory,
stochastic integration, and boundary value problems from PDEs among others.

In brief summary, a reproducing kernel Hilbert space consists of two things: a
Hilbert space of functions $f$ on a set $X$, and a reproducing kernel $k$,
i.e., a complex valued function $k$ on $X\times X$ such that for every $x$ in
$X$, the function $k(\cdot,x)$ is in and reproduces the value $f(x)$ from the
inner product $<f,k(\cdot,x)>$ in $\mathcal{H}$ so the formula
\[
f(x)=<f,k(\cdot,x)>
\]
holds for all $x$ in $X$. Moreover, there is a set of axioms for a function
$k$ in two variables that characterizes precisely when it determines a
reproducing kernel Hilbert space. And conversely there are necessary and
sufficient conditions that apply to Hilbert spaces $\mathcal{H}$ and decide
when $\mathcal{H}$ is a reproducing kernel Hilbert space. Here we shall
restrict these \textquotedblleft reproducing\textquotedblright\ axioms and
obtain instead a smaller class of reproducing kernel Hilbert spaces. We add
two additional axioms: Firstly, we will be reproducing not the values
themselves of the functions $f$ in $\mathcal{H}$, but rather the differences
$f(x)-f(y)$ for all pairs of points in $X$; and secondly we will impose one
additional axiom to the effect that the Dirac mass at $x$ is contained in
$\mathcal{H}$ for all $x$ in $\mathcal{H}$. In more precise form, the axioms
are as follows:

\begin{enumerate}
\item[(i)] For all $x,y\in X,~\exists w_{x,y}\in\mathcal{H}$ such that
$f\left(  x\right)  -f\left(  y\right)  =\left\langle f,w_{x,y}\right\rangle
$; and

\item[(ii)] For all $x\in X$, we have $\delta_{x}\in\mathcal{H}$.
\end{enumerate}

\textit{Quantum states} in physics are represented by norm-one vectors $v$ in
some Hilbert space $\mathcal{H}$, i.e., $\left\Vert v\right\Vert
_{\mathcal{H}}=1$. Hence the significance of assumption (ii) is to allow us to
\textquotedblleft place\textquotedblright\ quantum states on the points in
some prescribed set $X$ which allows a reproducing kernel-Hilbert space
$\mathcal{H}$, subject to condition (ii): If $x\in X$, then the corresponding
quantum state is $\left\Vert \delta_{x}\right\Vert _{\mathcal{H}}^{-1}%
\delta_{x}$; and the transition probability $x\longmapsto y$ is
\[
p_{x,y}\text{:}=\left\Vert \delta_{x}\right\Vert _{\mathcal{H}}^{-1}\left\Vert
\delta_{y}\right\Vert _{\mathcal{H}}^{-1}\left\vert \left\langle \delta
_{x},\delta_{y}\right\rangle _{\mathcal{H}}\right\vert \text{.}%
\]

When these two additional conditions (i)--(ii) are satisfied, we say that
$\mathcal{H}$ is a \textit{relative reproducing kernel Hilbert space}. It is
known that every weighted graph (the infinite case is of main concern here)
induces a relative reproducing kernel Hilbert space, and an associated graph
Laplacian. A main result in section \ref{ComputingSpaces} below is that the
converse holds: Given a relative reproducing kernel Hilbert space
$\mathcal{H}$ on a set $X$, it is then possible in a canonical way to
construct a weighted graph $G$ such that $X$ is the set of vertices in $G$,
and such that its energy Hilbert space coincides with $\mathcal{H}$ itself. In
our construction, the surprise is that the edges in $G$ as well as the weights
on the edges may be built directly from only the Hilbert space axioms defining
the initially given relative reproducing kernel Hilbert space. Since this
includes all infinite graphs of electrical resistors and their potential
theory (boundaries, harmonic functions, and graph Laplacians) the result has
applications to these fields, and it serves to unify diverse branches in a
vast research area.

\section{Other Applications\label{OtherApps}}

One additional application of our \textit{relative} reproducing kernel-Hilbert
spaces to infinite graphs $G$ entails the concept of \textquotedblleft
graph-boundary.\textquotedblright\ This is part of the study of discrete
dynamical systems and their harmonic analysis, i.e., following infinite paths
in the vertex set of $G$, and computing probabilities of sets of infinite paths.

While there is already a substantial literature on \textquotedblleft
boundaries\textquotedblright\ in the case of \textit{bounded} harmonic
functions on infinite weighted graphs $(G,w)$, our present setting has a quite
different flavor. We are concerned with harmonic functions $h$ of
\textit{finite energy}, and our reproducing kernel Hilbert spaces are chosen
such as to make this precise, as well as serving as a computational device. An
important technical point is that these \textquotedblleft finite-energy
Hilbert spaces\textquotedblright\ do not come equipped with an \textit{a
priori} realization as $L^{2}$-spaces.

This fact further explains why the resulting boundary theory is somewhat more
subtle than is the better known and better understood theory for the case of
bounded harmonic functions. Moreover there does not appear to be a direct way
of comparing the two \textquotedblleft boundary theories.\textquotedblright

There are some good intuitive reasons why stochastic integrals should
\textquotedblleft have something to\textquotedblright\ do with boundaries and
finite energy for infinite weighted graphs $(G,w)$;-- indeed be a crucial part
of this theory. Indeed, a fixed choice of weights on edges in $G$ (for example
conductance numbers) yields probabilities for a random walk. Going to the
\textquotedblleft boundary\textquotedblright\ for $(G,w)$ involves a subtle
notion of limit, and it is a well known principle that suitable limits of
random walk yield Brownian motion realized in $L^{2}$-spaces of global
measures (e.g., Wiener measure), and so corresponding to the stochastic nature
of Brownian motion.

The discreteness of vertex sets in infinite graphs, has a quantum aspect as
well \cite{CaPi08}, \cite{KlPa08}. It enters when inner products from a chosen
reproducing kernel-Hilbert space is used in encoding transition probabilities,
i.e., computing a transition between two vertices in $G$ as the absolute value
of the inner product of the corresponding Dirac-delta functions. Hence,
vertices in $G$ play the role of quantum states.

Let $\mathcal{H}$ be a reproducing kernel Hilbert space of functions on some
fixed set $X$; we assume properties (i)--(ii) above. There is then a dense
linear subspace $\mathcal{D}\subseteq\mathcal{H}$, and a hermitian operator
$\Delta\,$:$\,\mathcal{D}$ $\rightarrow\mathcal{H}$ determined by
\begin{equation}
\left(  \Delta u\right)  \left(  x\right)  \text{:}=\left\langle \delta
_{x},u\right\rangle ,~\forall u\in\mathcal{D}\text{,} \label{Eq2.1}%
\end{equation}
where $\left\langle \cdot,\cdot\right\rangle $:$=\left\langle \cdot
,\cdot\right\rangle _{\mathcal{H}}$ refers to the inner product in
$\mathcal{H}$.

\begin{definition}
\label{Def2.1}Let $\mathcal{H}$, $X,$ $\Delta$ be as described above\emph{;}
and let $x_{0}\in X$ be given. A vector $w\left(  =w_{x_{0}}\right)  $ is said
to be a \emph{monopole} if
\begin{equation}
\left\langle w,\Delta u\right\rangle =\left\langle \delta_{x_{0}%
},u\right\rangle \text{ for all }u\in\mathcal{D}\text{.} \label{Eq2.2}%
\end{equation}

\end{definition}

\noindent(Contrast (\ref{Eq2.2}) with condition (i) above. A function
$w_{x,y}\in\mathcal{H}$ satisfying (i) is called a \textit{bipole}. We will
see that bipoles always exist, while monopoles do not.)

\begin{example}
\label{Ex2.2}Consider functions $u$ on the integers $\mathbb{Z}$ subject to
the condition
\begin{equation}
\left\Vert u\right\Vert ^{2}=\sum_{x\in\mathbb{Z}}\left\vert u\left(
x\right)  -u\left(  x+1\right)  \right\vert ^{2}<\infty\text{.} \label{Eq2.3}%
\end{equation}
Moding out with the constant functions on $\mathbb{Z}$, note that $\left\Vert
\cdot\right\Vert $ in \emph{(}\ref{Eq2.3}\emph{)} is then a Hilbert norm. The
corresponding Hilbert space will be denoted $\mathcal{H}$.
\end{example}

It is convenient to realize $\mathcal{H}$ via the following Fourier series
representation
\begin{equation}
\tilde{u}\left(  \theta\right)  \text{:}=\sum_{x\in\mathbb{Z}}u\left(
x\right)  e^{ix\cdot\theta}\text{,} \label{Eq2.4}%
\end{equation}
i.e., a $2\pi$-periodic function. Note that the same construction works
\textit{mutatis mutandis} in $d$ variables for $d>1$.

\begin{lemma}
\label{Lem2.3}A $2\pi$-periodic function $\tilde{u}$ as in \emph{(}%
\ref{Eq2.4}\emph{)} represents an $u\in\mathcal{H}$ if and only if
\[
\sin\left(  \frac{\theta}{2}\right)  \tilde{u}\left(  \theta\right)  \in
L^{2}\left(  -\pi,\pi\right)  ;
\]
and in that case
\begin{equation}
\left\Vert u\right\Vert ^{2}=\frac{2}{\pi}\int_{-\pi}^{\pi}\sin^{2}\left(
\frac{\theta}{2}\right)  \left\vert \tilde{u}\left(  \theta\right)
\right\vert ^{2}~d\theta. \label{Eq2.5}%
\end{equation}

\end{lemma}

\begin{remark}
\label{Rem2.4}Note that the constant function $u_{1}\equiv1$ on $\mathbb{Z}$
does not contribute to \emph{(}\ref{Eq2.3}\emph{);} and as a result $\tilde
{u}_{1}\left(  \theta\right)  =\delta\left(  \theta-0\right)  $ does not
contribute to \emph{(}\ref{Eq2.5}\emph{)}. The last fact can be verified directly.
\end{remark}

\begin{proof}
[Proof of Lemma \ref{Lem2.3}]For the RHS in (\ref{Eq2.3}) we have
\[
\left\vert \left(  \tilde{u}\left(  \cdot\right)  -\widetilde{u\left(
\cdot+1\right)  }\right)  \left(  \theta\right)  \right\vert =\left\vert
\left(  1-e^{-i\theta}\right)  \tilde{u}\left(  \theta\right)  \right\vert
=2\left\vert \sin\left(  \frac{\theta}{2}\right)  \tilde{u}\left(
\theta\right)  \right\vert ;
\]
and the conclusion in (\ref{Eq2.5}) now follows from Parseval's formula for
Fourier series.
\end{proof}

\begin{lemma}
\label{Lem2.5}The Hilbert space $\left(  \mathcal{H},\mathbb{Z}\right)  $ in
Example \ref{Ex2.2} has dipoles, but not monopoles.
\end{lemma}

\begin{proof}
Let $x,y\in\mathbb{Z}$. We may assume without loss that $0\leq y<x$. Now set
\begin{equation}
w_{x,y}\left(  n\right)  \text{:}=\left\{
\begin{array}
[c]{l}%
\,0,~\text{if }\left\vert n\right\vert \leq y\\
\left\vert n\right\vert -y,~\text{if }y<\left\vert n\right\vert \leq x\\
\;x-y,~\text{if }x<\left\vert n\right\vert \text{.}%
\end{array}
\right.  \label{Eq2.6}%
\end{equation}
Then the reproducing formula holds, i.e., we have (i):
\begin{equation}
\left\langle w_{x,y},u\right\rangle _{\mathcal{H}}=u\left(  x\right)
-u\left(  y\right)  ,~\forall u\in\mathcal{H}\text{.} \label{Eq2.7}%
\end{equation}
Setting $v_{x}$:$=w_{x,0}$, we get
\begin{equation}
\left\langle v_{x},v_{y}\right\rangle _{\mathcal{H}}=\left\vert x\right\vert
\wedge\left\vert y\right\vert . \label{Eq2.8}%
\end{equation}

The fact that there are no monopoles follows from the observation that
\begin{equation}
\tilde{w}\left(  \theta\right)  =\frac{e^{ix\cdot\theta}}{4\sin^{2}\left(
\frac{\theta}{2}\right)  } \label{Eq2.9}%
\end{equation}
does note satisfy the finiteness condition in (\ref{Eq2.5}).
\end{proof}

\begin{remark}
\label{Rem2.6}Set $X=\mathbb{Z}$, and let $\mathcal{H}$ be the Hilbert space
in Example \ref{Ex2.2} of functions $f:\mathbb{Z}\rightarrow\mathbb{C}$,
modulo the constant functions, such that
\[
\left\Vert f\right\Vert _{\mathcal{H}}^{2}=\sum_{x\in\mathbb{Z}}\left\vert
f\left(  x\right)  -f\left(  x+1\right)  \right\vert ^{2}<\infty\text{.}%
\]
A computation reveals the following three facts \emph{(}details in section
\ref{ComputingSpaces} below\emph{):}

\begin{itemize}
\item[\emph{(}a\emph{)}] For all $x\in\mathbb{Z}\,\diagdown\left(  0\right)
$, there is a $v_{x}\in\mathcal{H}$ such that
\[
\left\langle v_{x},f\right\rangle _{\mathcal{H}}=f\left(  x\right)  -f\left(
0\right)  \text{ holds for all }f\in\mathcal{H}.
\]

\item[\emph{(}b\emph{)}] There is \emph{no} $w\in\mathcal{H}$ such that
\[
\left\langle w,f\right\rangle _{\mathcal{H}}=f\left(  x\right)  \text{ holds
for all }f\in\mathcal{H}.
\]

\item[\emph{(}c\emph{)}] Functions in $\mathcal{H}$ may be \emph{unbounded:}
Take for example $f\left(  x\right)  $\emph{:}$=\log\left(  1+\left\vert
x\right\vert \right)  ,$ defined for all $x\in\mathbb{Z}.$
\end{itemize}
\end{remark}

A glance at the defining conditions (i) -- (ii) for \textquotedblleft relative
reproducing kernel Hilbert spaces\textquotedblright\ suggests applications to
\textquotedblleft boundaries\textquotedblright\ of infinite discrete
configurations, such as infinite weighted graphs.

One of the aims of our paper is to study precisely this: The introduction of a
suitable reproducing kernel Hilbert space into the analysis of an infinite
configuration $X$ leads to an associated \textquotedblleft
boundary\textquotedblright, i.e., to a compactification of $X$, so the
boundary consisting of the points in the compactification not already in $X$;
hence notions not present in the finite case; see especially our operator
theoretic formulation of \textquotedblleft boundary\textquotedblright\ in
section \ref{Extensions} below.

The study of \textquotedblleft boundary terms\textquotedblright\ is central to
our approach. In contrast to other related but different notions in the
literature of \textquotedblleft boundary\textquotedblright\ for random walks,
we employ here tools intrinsic to unbounded operators with dense domain in
Hilbert space. To start this, we must first, for a given infinite
configuration $X$, identify the \textquotedblleft right\textquotedblright%
\ Hilbert space; see sections \ref{Extensions} and \ref{PairsInDuality} below.
Our boundary \textquotedblleft$\operatorname*{bd}X$\textquotedblright%
\ (section \ref{ConcludingRem}) is comparable to, but different from, other
boundaries in the literature.

\section{Extensions of Unbounded Operators\label{Extensions}}

\begin{definition}
\label{Def3.1}~\newline

\begin{itemize}
\item $\mathcal{H}$\emph{:} some given complex Hilbert space with fixed inner
product $<\cdot,\cdot>$ and norm $\left\Vert \cdot\right\Vert $.

\item $\mathcal{D}\subset\mathcal{H}$\emph{:} some given dense linear subspace
in $\mathcal{H}$.

\item $\Delta:\mathcal{D}\rightarrow\mathcal{H}$\emph{:} a given linear
operator\emph{;} typically unbounded.
\end{itemize}
\end{definition}

We say that $\Delta$ is \textit{Hermitian} iff
\begin{equation}
\left\langle u,\Delta v\right\rangle =\left\langle \Delta u,v\right\rangle
,~\forall u,v\in\mathcal{D}\text{;} \label{Eq3.1}%
\end{equation}
and we say that $\mathcal{D}$ is the \textit{domain} of $\Delta$; written

\begin{itemize}
\item $\operatorname*{dom}\left(  \Delta\right)  $:$=\mathcal{D}$.
\end{itemize}

The \textit{adjoint }operator $\Delta^{\ast}$ is defined as follows: Let
\[
\operatorname*{dom}\left(  \Delta^{\ast}\right)  \text{:}=\left\{  \psi
\in\mathcal{H}|\text{ such that }\exists\mathcal{C}<\infty\text{ with
}\left\vert \left\langle \psi,\Delta v\right\rangle \right\vert \subseteq
\mathcal{C}\left\Vert v\right\Vert ,~\forall v\in\mathcal{D}\right\}  \text{.}%
\]

If $\psi\in\operatorname*{dom}\left(  \Delta^{\ast}\right)  $, then by Riesz'
lemma, there is a unique $w\in\mathcal{H}$ such that
\begin{equation}
\left\langle \psi,\Delta v\right\rangle =\left\langle w,v\right\rangle
,~\forall v\in\mathcal{D}\text{;} \label{Eq3.2}%
\end{equation}
and we set $\Delta^{\ast}\psi$:$=w$.

The graph $G$ of an operator $\Delta$ is defined by
\begin{equation}
G\left(  \Delta\right)  \text{:}=\left\{  \left(
\genfrac{}{}{0pt}{}{v}{\Delta v}%
\right)  |v\in\operatorname*{dom}\left(  \Delta\right)  \right\}
\subseteq\mathcal{H}\times\mathcal{H}\text{.} \label{Eq3.3}%
\end{equation}

If $\Delta$ is hermitian, there is a closed hermitian operator $\Delta
^{\operatorname*{clo}}$ such that%
\begin{equation}
G\left(  \Delta^{\operatorname*{clo}}\right)  \text{:}=G\left(  \Delta\right)
^{\left\Vert \cdot\right\Vert \times\left\Vert \cdot\right\Vert -closure}%
\text{.} \label{Eq3.4}%
\end{equation}

One checks that
\begin{equation}
\left(  \Delta^{\operatorname*{clo}}\right)  ^{\ast}=\Delta^{\ast}\text{.}
\label{Eq3.5}%
\end{equation}

For a pair of operators $\Delta_{1}$ and $\Delta_{2}$ we say that
\begin{equation}
\Delta_{1}\subseteq\Delta_{2}\underset{\operatorname*{Def}%
\nolimits^{\underline{n}.}}{\Leftrightarrow}G\left(  \Delta_{1}\right)
\subseteq G\left(  \Delta_{2}\right)  \text{.} \label{Eq3.6}%
\end{equation}

\begin{lemma}
\label{Lem3.2}Let $T$ be a hermitian extension of $\Delta$. Then the following
containments hold\emph{:}
\begin{equation}
\Delta\subseteq T\subseteq T^{\operatorname*{clo}}\subseteq T^{\ast}%
\subseteq\Delta^{\ast}\text{.} \label{Eq3.7}%
\end{equation}

\end{lemma}

\begin{corollary}
\label{Cor3.3}Let $T$ be a hermitian extension of $\Delta$\emph{;} then
\begin{equation}
\operatorname*{dom}\left(  T^{\ast}\right)  \subseteq\operatorname*{dom}%
\left(  \Delta^{\ast}\right)  \text{.} \label{Eq3.8}%
\end{equation}

\end{corollary}

\begin{proof}
Immediate from (\ref{Eq3.6}) and (\ref{Eq3.7}).
\end{proof}

We now turn to a specific family of hermitian extensions of a fixed densely
defined operator $\Delta$.

\begin{definition}
\label{Def3.4}Let $\Delta$ be a hermitian operator with dense domain
$\mathcal{D}$ in a Hilbert space $\mathcal{H}$. Let $\mathcal{C}$ be a closed
subspace in $\mathcal{H}$, and assume that
\begin{equation}
\mathcal{C}\subset\operatorname*{dom}\left(  \Delta^{\ast}\right)  \text{.}
\label{Eq3.9}%
\end{equation}

\end{definition}

On the space
\begin{equation}
\mathcal{D}+\mathcal{C}=\left\{  v+h|v\in\mathcal{D},~h\in\mathcal{C}\right\}
\label{Eq3.10}%
\end{equation}
set
\begin{equation}
\Delta_{\mathcal{C}}\left(  v+h\right)  :=\Delta v,~\text{for }v\in
\mathcal{D}\text{ and }h\in\mathcal{C}. \label{Eq3.11}%
\end{equation}

\begin{lemma}
\label{Lem3.5}Let $\Delta$, $\mathcal{H}$, and $\mathcal{C}$ be as in the
definition. Then the following two conditions are equivalent\emph{:}

\begin{enumerate}
\item[(i)] $\Delta_{\mathcal{C}}$ in \emph{(}\ref{Eq3.11}\emph{)} is a well
defined hermitian extension operator\emph{;} and

\item[(ii)] $\mathcal{C}\subseteq\ker\left(  \Delta^{\ast}\right)  $.
\end{enumerate}
\end{lemma}

\begin{proof}
(i)$\Rightarrow$(ii). From (i) we conclude that the following implication
holds:
\begin{equation}
\left(  v\in\mathcal{D},~h\in\mathcal{C},~v+h=0\right)  \Rightarrow\Delta v=0.
\label{Eq3.12}%
\end{equation}
Now use (\ref{Eq3.9}), and apply $\Delta^{\ast}$ to $v+h$ in (\ref{Eq3.12}):
We get
\[
0=\Delta^{\ast}\left(  v+h\right)  =\Delta^{\ast}v+\Delta^{\ast}h=\Delta
v+\Delta^{\ast}h=\Delta^{\ast}h;
\]
so $h\in\ker\left(  \Delta^{\ast}\right)  $. This applies to all
$h\in\mathcal{C}$ so (ii) holds.

(ii)$\Rightarrow$(i). Assume (ii). We must then prove the implication
(\ref{Eq3.12}). Then it follows that $\Delta_{\mathcal{C}}$ is a well defined
extension operator. If $v\in\mathcal{D}$, $h\in\mathcal{C}$, and $v+h=0$,
then
\[
\Delta^{\ast}\left(  v+h\right)  =0=\Delta^{\ast}v+\Delta^{\ast}h=\Delta v
\]
since $h\in\ker\left(  \Delta^{\ast}\right)  $. Hence $\Delta v=0$ which
proves (\ref{Eq3.12}).

To prove that $\Delta_{\mathcal{C}}$ is hermitian consider vectors $\psi_{i}%
$:$=v_{i}+h_{i}$, $v_{i}\in\mathcal{D}$, $h_{i}\in\mathcal{C}$, $i=1,2$.

Then
\begin{align*}
\left\langle \Delta_{\mathcal{C}}\psi_{1},\psi_{2}\right\rangle  &
=\left\langle \Delta v_{1},v_{2}+h_{2}\right\rangle \\
&  =\left\langle \Delta v_{1},v_{2}\right\rangle +\left\langle \Delta
v_{1},h_{2}\right\rangle \\
&  =\left\langle v_{1},\Delta v_{2}\right\rangle +\left\langle v_{1}%
,\Delta^{\ast}h_{2}\right\rangle \\
&  =\left\langle v_{1},\Delta v_{2}\right\rangle \\
&  =\left\langle v_{1}+h_{1},\Delta v_{2}\right\rangle \\
&  =\left\langle \psi_{1},\Delta_{\mathcal{C}}\psi_{2}\right\rangle ,
\end{align*}
which is the desired conclusion; in other words, $\Delta_{\mathcal{C}}$ is a
hermitian extension operator.
\end{proof}

\begin{theorem}
\label{Theo3.6}Let $\Delta$ be a hermitian operator with dense domain
$\mathcal{D}$ in a Hilbert space, and let $\mathcal{C}$ be a closed subspace
such that $\mathcal{C}$ $\subset\ker\left(  \Delta^{\ast}\right)  $. Let
$\Delta_{\mathcal{C}}$ be the corresponding hermitian extension operator.

Then
\begin{equation}
\operatorname*{dom}\left(  \Delta_{\mathcal{C}}^{\ast}\right)  =\left\{
\psi\in\operatorname*{dom}\left(  \Delta^{\ast}\right)  |\Delta^{\ast}\psi
\in\mathcal{H}\ominus\mathcal{C}\right\}  \label{Eq3.13}%
\end{equation}
where
\begin{equation}
\mathcal{H}\ominus\mathcal{C}\text{:=}\left\{  \varphi\in\mathcal{H}%
|\left\langle \varphi,h\right\rangle =0,\forall h\in\mathcal{C}\right\}
\text{.} \label{Eq3.14}%
\end{equation}

\end{theorem}

\begin{proof}
Whenever $\mathcal{C}$ is a closed subspace, the corresponding orthogonal
projection will be denoted $P_{\mathcal{C}}$. Recall $P_{\mathcal{C}}$
satisfies
\begin{equation}
P_{\mathcal{C}}=P_{\mathcal{C}}^{\ast}=P_{\mathcal{C}}^{2},\text{ and}
\label{Eq3.15}%
\end{equation}%
\begin{equation}
P_{\mathcal{C}}\mathcal{H}=\mathcal{C}\text{.} \label{Eq3.16}%
\end{equation}

Set $P_{\mathcal{C}}^{\bot}=I-P_{\mathcal{C}}$; then $P_{\mathcal{C}}^{\bot}$
is the projection onto $\mathcal{H}\ominus\mathcal{C}$.

For a one-dimensional subspace spanned by a single vector $h\not =0$, we have
\begin{equation}
P_{h}v=\left\Vert h\right\Vert ^{-2}\left\langle h,v\right\rangle h,~\forall
v\in\mathcal{H}\text{.} \label{Eq3.17}%
\end{equation}

We now turn to the proof of (\ref{Eq3.13}). First this inclusion:

$\quad\left(  \supseteq\right)  $ (Easy direction!) So let $\psi
\in\operatorname*{dom}\left(  \Delta^{\ast}\right)  $, and assume that
$\Delta^{\ast}\psi\in\mathcal{C}^{\bot}$. Then we get the following estimate:
\begin{align*}
\left\vert \left\langle \psi,\Delta_{\mathcal{C}}\left(  v+h\right)
\right\rangle \right\vert  &  =\left\vert \left\langle \psi,\Delta
v\right\rangle \right\vert \\
&  =\left\vert \left\langle \Delta^{\ast}\psi,v\right\rangle \right\vert \\
&  =\left\vert \left\langle \Delta^{\ast}\psi,v+h\right\rangle \right\vert \\
&  \leq\left\Vert \Delta^{\ast}\psi\right\Vert \cdot\left\Vert v+h\right\Vert
\end{align*}
valid for all $v\in\mathcal{D}$, and all $h\in\mathcal{C}$.

We conclude from Definition \ref{Def3.1} that $\psi\in\operatorname*{dom}%
\left(  \Delta_{\mathcal{C}}^{\ast}\right)  $.

\quad$\left(  \subseteq\right)  $ Conversely, if some fixed vector $\psi$ is
in the $\operatorname*{dom}\left(  \Delta_{\mathcal{C}}^{\ast}\right)  $, then
there is a constant $C<\infty$ such that $\left\vert \left\langle \psi
,\Delta_{\mathcal{C}}\left(  v+h\right)  \right\rangle \right\vert \leq
C\left\Vert v+h\right\Vert $ for all $v\in\mathcal{D}$ $and$ $h\in\mathcal{C}$.

Since $\Delta_{\mathcal{C}}\left(  v+h\right)  =\Delta v$, we get
\begin{equation}
\left\vert \left\langle \psi,\Delta v\right\rangle \right\vert ^{2}\leq
C^{2}\left\Vert v+h\right\Vert ^{2}=C^{2}\left(  \left\Vert v\right\Vert
^{2}+2\operatorname*{Re}\left\langle v,h\right\rangle +\left\Vert h\right\Vert
^{2}\right)  \text{.} \label{Eq3.18}%
\end{equation}

Now replacing $h$ with $\lambda h$ for $\lambda\in\mathbb{C}$, we arrive at
the following estimate; essentially an application of Schwarz' inequality:

As a result we get the following estimate:
\[
\left\Vert h\right\Vert ^{2}\left\vert \left\langle \psi,\Delta v\right\rangle
\right\vert ^{2}\leq C^{2}\cdot\left(  \left\Vert h\right\Vert ^{2}%
\cdot\left\Vert v\right\Vert ^{2}-\left\vert \left\langle h,v\right\rangle
\right\vert ^{2}\right)  ,
\]
valid for all $v\in\mathcal{D}$ and $h\in\mathcal{C}$. Or equivalently:
\[
\left\vert \left\langle \psi,\Delta v\right\rangle \right\vert ^{2}\leq
C^{2}\cdot\left(  \left\Vert v\right\Vert ^{2}-\frac{\left\vert \left\langle
h,v\right\rangle \right\vert ^{2}}{\left\Vert h\right\Vert ^{2}}\right)  .
\]

Introducing the rank-one projection $P_{h}$ this then reads as follows:
\[
\left\vert \left\langle \psi,\Delta v\right\rangle \right\vert ^{2}\leq
C^{2}\cdot\left(  \left\Vert v\right\Vert ^{2}-\left\Vert P_{h}v\right\Vert
^{2}\right)  \text{;}%
\]
or equivalently:
\begin{equation}
\left\vert \left\langle \psi,\Delta v\right\rangle \right\vert ^{2}\leq
C^{2}\cdot\left\Vert P_{h}^{\bot}v\right\Vert ^{2}\text{;} \label{Eq3.19}%
\end{equation}
See equation \ref{Eq3.17}.

An application of Riesz' theorem then yields a vector $\varphi^{\ast}%
\in\left\{  h\right\}  ^{\bot}$ such that
\begin{equation}
\left\langle \psi,\Delta v\right\rangle =\left\langle \varphi^{\ast}%
,P_{h}^{\bot}v\right\rangle \label{Eq3.20}%
\end{equation}
valid for all $v\in\mathcal{D}$.

But $\left\langle \varphi^{\ast},P_{h}^{\bot}v\right\rangle =\left\langle
\varphi^{\ast},v\right\rangle $, and we conclude that
\begin{equation}
\left\langle \psi,\Delta v\right\rangle =\left\langle \varphi^{\ast
},v\right\rangle \label{Eq3.21}%
\end{equation}
for all $v\in\mathcal{D}$. From (\ref{Eq3.21}), and Definition \ref{Def3.1},
we conclude that $\psi\in\operatorname*{dom}\left(  \Delta^{\ast}\right)  $,
and that
\[
\left\langle \Delta^{\ast}\psi-\varphi^{\ast},v\right\rangle =0~\text{for all
}v\in\mathcal{D}.
\]
Since $\mathcal{D}$ is dense in $\mathcal{H}$, we get
\[
\Delta^{\ast}\psi=\varphi^{\ast}\in\left\{  h\right\}  ^{\bot},
\]
and therefore
\[
\Delta^{\ast}\psi\in\bigcap\nolimits_{h\in\mathcal{C}}\left\{  h\right\}
^{\bot}=\mathcal{C}^{\bot}=\mathcal{H}\ominus\mathcal{C}.
\]

\end{proof}

\section{Pairs of Hermitian operators in duality\label{PairsInDuality}}

In Lemma \ref{Lem3.5} we introduced the following fundamental properties for a
pair $\left(  \Delta,\mathcal{C}\right)  $ where $\Delta$ is a given Hermitian
operator with dense domain $\mathcal{D}$ in a fixed Hilbert space
$\mathcal{H}$; and where $\mathcal{C}$ is a closed subspace in $\mathcal{H}$.

\begin{definition}
\label{Def4.1}Let $\left(  \Delta,\mathcal{C}\right)  $ be a pair as described
above, and let $\mathcal{H}$ be the ambient Hilbert space. We say the $\left(
\Delta,\mathcal{C}\right)  $ is a \emph{duality pair} iff the inclusion
\begin{equation}
\mathcal{C}\subseteq\ker\left(  \Delta^{\ast}\right)  \label{Eq4.1}%
\end{equation}
holds.

Let $R\left(  \Delta\right)  =\left\{  \Delta v|v\in\mathcal{D}\right\}  $ be
the range of $\Delta$, and
\begin{equation}
R\left(  \Delta\right)  ^{\operatorname*{clo}}=R\left(  \Delta\right)
^{\bot\bot} \label{Eq4.2}%
\end{equation}
the norm closure in $\mathcal{H}$.
\end{definition}

\begin{lemma}
\label{Lem4.2}For a pair $\left(  \Delta,\mathcal{C}\right)  $ in
$\mathcal{H}$, the following conditions are equivalent\emph{:}

\begin{enumerate}
\item[(i)] $\left(  \Delta,\mathcal{C}\right)  $ is a duality pair\emph{;} and

\item[(ii)] $R\left(  \Delta\right)  ^{\operatorname*{clo}}\subseteq
\mathcal{H}\ominus\mathcal{C}$.
\end{enumerate}
\end{lemma}

\begin{proof}
(i)$\Rightarrow$(ii). Given (\ref{Eq4.1}), we may take ortho-complements, and
\begin{equation}
\left(  \ker\left(  \Delta^{\ast}\right)  \right)  ^{\bot}\subseteq
\mathcal{C}^{\bot} \label{Eq4.3}%
\end{equation}
The desired (ii) now follows from
\[
R\left(  \Delta\right)  ^{\operatorname*{clo}}=R\left(  \Delta\right)
^{\bot\bot}=\left(  \ker\left(  \Delta^{\ast}\right)  \right)  ^{\bot
}\text{and }\mathcal{C}^{\bot}=\mathcal{H}\ominus\mathcal{C}.
\]

(ii)$\Rightarrow$(i). The above argument works in reverse: Take perpendicular
on both sides in (\ref{Eq4.3}), and note that the containment reverses, so
(ii) implies (i).
\end{proof}

The following family of reproducing kernel Hilbert spaces includes duality
pairs. This in turn includes all graph-Laplacians on infinite weighted graphs,
as we will show.

\begin{definition}
\label{Def4.3}Let $X$ be a set. Pick some $o\in X$, and set $X^{\ast}$%
\emph{:}$=X\diagdown\left\{  0\right\}  $. A Hilbert space $\mathcal{H}$ is
said to be a \emph{reproducing kernel Hilbert space with base-point} if there
is a function
\begin{equation}
k\text{\emph{:\thinspace}}X\times X^{\ast}\rightarrow\mathbb{C} \label{Eq4.4}%
\end{equation}
such that
\begin{equation}
v_{x}\left(  \cdot\right)  \text{\emph{:}}=k\left(  \cdot,x\right)
\in\mathcal{H},~\forall x\in X^{\ast}; \label{Eq4.5}%
\end{equation}%
\begin{equation}
\left\langle v_{x},\,f\right\rangle _{\mathcal{H}}=f\left(  x\right)
-f\left(  o\right)  ,~\forall f\in\mathcal{H},~\forall x\in X^{\ast}\text{.}
\label{Eq4.6}%
\end{equation}
In particular, $\mathcal{H}$ is a space of functions on $X$. The inner product
in $\mathcal{H}$ is denoted $\left\langle \cdot,\cdot\right\rangle
_{\mathcal{H}}$ or simply $\left\langle \cdot,\cdot\right\rangle $.
\end{definition}

In addition, we require
\begin{equation}
\text{closed span}\left\{  v_{x}|x\in X^{\ast}\right\}  =\mathcal{H};\text{
and} \label{Eq4.7}%
\end{equation}%
\begin{equation}
\left\{  \delta_{x}|x\in X\right\}  \subset\mathcal{H}\text{.} \label{Eq4.8}%
\end{equation}
Hence
\begin{equation}
\delta_{x}\left(  y\right)  =\left\{
\begin{array}
[c]{ll}%
1 & \text{if }y=x\\
0 & \text{if }y\not =x\text{ in }X;
\end{array}
\right.  \label{Eq4.9}%
\end{equation}
i.e., the Dirac-functions on $X$.

\begin{remark}
\label{Rem4.4}\emph{(}a\emph{)} Because of \emph{(}\ref{Eq4.6}\emph{)},
$\mathcal{H}$ is really a space of functions modulo the constant functions.

\emph{(}b\emph{)} Not all reproducing kernel Hilbert spaces have property
\emph{(}\ref{Eq4.8}\emph{):} Take for example $X$:$=\left[  0,1\right]
,~o=0$,
\begin{equation}
k\left(  x,y\right)  \text{\emph{:}}=x\wedge y\text{,} \label{Eq4.10}%
\end{equation}
i.e., the smallest two numbers.
\end{remark}

Let $\mathcal{H}$ be the space of measurable functions $f$ on $X$ such that
the distribution derivative $f^{\prime}=\frac{df}{dx}$ is in $L^{2}\left(
0,1\right)  $. Set
\begin{equation}
\left\Vert f\right\Vert _{\mathcal{H}}^{2}\text{:}=\int\limits_{0}%
^{1}\left\vert f^{\prime}\left(  x\right)  \right\vert ^{2}~dx\text{.}
\label{Eq4.11}%
\end{equation}

It is easy to check then that conditions (\ref{Eq4.5})--(\ref{Eq4.7}) will be
satisfied; but that (\ref{Eq4.8}) will \textit{not} hold.

On the other hand, energy Hilbert spaces for weighted graphs will satisfy
(\ref{Eq4.8}). Specifically, let $\left(  G,c\right)  =\left(  G^{0}%
,G^{1},c\right)  $ be an (infinite) weighted graph, i.e.,

\begin{itemize}
\item $G^{0}=$ the vertex set (discrete);

\item $G^{1}\subset G^{0}\times G^{0}$ is the set of edges in $G$;

\item $c\,$:$\,G^{1}\rightarrow\mathbb{R}$ a fixed weight function such that
$c\left(  xy\right)  =c\left(  yx\right)  $ for all $\left(  xy\right)  \in
G^{1}$.
\end{itemize}

For functions $u$ and $v$ on $G^{0}$ set
\begin{equation}
\left\langle u,v\right\rangle _{\mathcal{H}}\text{:}=\frac{1}{2}\underset{%
\genfrac{}{}{0pt}{}{x,y}{\text{s.t. }\left(  xy\right)  \in G^{1}}%
}{\sum\sum}c\left(  xy\right)  \left(  \overline{u\left(  x\right)
}-\overline{u\left(  y\right)  }\right)  \left(  v\left(  x\right)  -v\left(
y\right)  \right)  ; \label{Eq4.12}%
\end{equation}
and $\left\Vert u\right\Vert _{\mathcal{H}}^{2}=\left\langle u,u\right\rangle
_{\mathcal{H}}$. (We choose our inner product to be linear in the second variable.)

The Hilbert space $\mathcal{H}$ consists of all functions $u$ such that
\begin{equation}
\left\Vert u\right\Vert _{\mathcal{H}}^{2}=\underset{%
\genfrac{}{}{0pt}{}{x,y}{\left(  xy\right)  \in G^{1}}%
}{\sum\sum}c\left(  xy\right)  \left\vert u\left(  x\right)  -u\left(
y\right)  \right\vert ^{2}<\infty\text{.} \label{Eq4.13}%
\end{equation}

We proved in \cite{JoPe08} that $\mathcal{H}$ is a reproducing kernel Hilbert
space with base-point; in particular, if $o\in G^{0}$ is chosen, then
conditions (\ref{Eq4.5})--(\ref{Eq4.8}) are satisfied.

Here we shall include (\ref{Eq4.8}) as part of our definition.

More precisely:

\begin{proposition}
\label{Prop4.5}Let $\left(  G,c\right)  $ be a weighted graph with energy
Hilbert space $\mathcal{H}=\mathcal{H}_{E}$.

Pick a base-point $o\in G^{0}$, and let $\left(  v_{x}\right)  _{x\in
G^{0}\diagdown\left(  0\right)  }$ be the family \emph{(dipoles)}\thinspace
from \emph{(}\ref{Eq4.6}\emph{)}.

Suppose, for all $x\in G^{0}$,
\begin{equation}
c\left(  x\right)  \text{\emph{:}}=\sum_{%
\genfrac{}{}{0pt}{}{y\text{, such that}}{\left(  xy\right)  \in G^{1}}%
}c\left(  xy\right)  <\infty; \label{Eq4.14}%
\end{equation}
then \emph{(}\ref{Eq4.8}\emph{)} holds, and
\begin{equation}
\delta_{x}=c\left(  x\right)  v_{x}-\sum_{\left(  xy\right)  \in G^{1}%
}c\left(  xy\right)  v_{y}\text{.} \label{Eq4.15}%
\end{equation}

\end{proposition}

\begin{proof}
By a direct computation, using (\ref{Eq4.13}) and (\ref{Eq4.14}), we get
\begin{equation}
\left\Vert \delta_{x}\right\Vert _{\mathcal{H}}^{2}=\underset{%
\genfrac{}{}{0pt}{}{y\text{, such that}}{\left(  xy\right)  \in G^{1}}%
}{\sum}c\left(  xy\right)  =c\left(  x\right)  <\infty\text{.} \label{Eq4.16}%
\end{equation}

\end{proof}

\begin{lemma}
\label{Lem4.Six}Let $\left(  \mathcal{H},X,\Delta\right)  $ be a reproducing
kernel system as outlined in section \ref{ReproKHspaces}. Let $x_{0}\in
X$\emph{;} then some $w_{0}\in\mathcal{H}$ is a monopole at $x_{0}$ if and
only if $w_{0}\in\operatorname*{dom}\left(  \Delta^{\ast}\right)  $ and
\begin{equation}
\Delta^{\ast}w_{0}=\delta_{x_{0}}\text{.} \label{Eq4.Seventeen}%
\end{equation}

\end{lemma}

\begin{proof}
If $w_{0}\in\mathcal{H}$ is a monopole at $x_{0}$, then (\ref{Eq2.2}) holds,
i.e., $\left\langle w_{0},\Delta u\right\rangle =\left\langle \delta_{x_{0}%
},u\right\rangle $ is satisfied for all $u\in\mathcal{D}$. Since $\left\vert
\left\langle \delta_{x_{0}},u\right\rangle \right\vert \leq\left\Vert
\delta_{x_{0}}\right\Vert \cdot\left\Vert u\right\Vert $ holds by Schwarz, it
follow that (\ref{Eq4.Seventeen}) is satisfied.

The argument for the converse implication is an application of Riesz' lemma to
the Hilbert space $\mathcal{H}$.
\end{proof}

\begin{remark}
\label{Rem4.6}It is not true in general that the truncated summations on the
R.H.S. in \emph{(}\ref{Eq4.15}\emph{)} converge in the norm \emph{(}%
\ref{Eq4.13}\emph{)} of $\mathcal{H}_{E}$. But it is if each $x\in G^{0}$ has
at most a \emph{finite} number of \emph{neighbors}. For pairs of points in
$G^{0}$, set
\begin{align}
x  &  \sim y\text{ iff there is an edge}\label{Eq4.17}\\
e  &  \in G^{1}\text{ with }e=\left(  xy\right)  \text{.}\nonumber
\end{align}

Set
\[
\operatorname*{Nbh}\nolimits_{G}\left(  x\right)  \text{:}=\left\{  y\in
G^{0}|y\sim x\right\}  \text{.}%
\]
We say that $G$ has finite degrees if
\begin{equation}
\#\operatorname*{Nbh}\nolimits_{G}\left(  x\right)  <\infty,~\forall x\in
G^{0}\text{.} \label{Eq4.18}%
\end{equation}

\end{remark}

\begin{theorem}
\label{Theo4.7}Let $\left(  \mathcal{H},X\right)  $ be a reproducing kernel
Hilbert space with base point $o$, and assume \emph{(}\ref{Eq4.8}\emph{)} is
satisfied. For $x\in X$, and $f\in\mathcal{H}$, set
\begin{equation}
\left(  \Delta f\right)  \left(  x\right)  \text{\emph{:}}=\left\langle
\delta_{x},f\right\rangle ; \label{Eq4.19}%
\end{equation}
then $\Delta$ is a hermitian operator with dense domain
\begin{equation}
\mathcal{D}_{V}\text{\emph{:}}=\operatorname*{span}\left\{  v_{x}|x\in
X^{\ast}\right\}  \text{.} \label{Eq4.20}%
\end{equation}

It satisfies$:$
\begin{equation}
\Delta v_{x}=\delta_{x}-\delta_{0};\text{ and} \label{Eq4.21}%
\end{equation}%
\begin{equation}
\left\langle u,\Delta u\right\rangle \geq0,~\forall u\in\mathcal{D}%
_{V}\text{.} \label{Eq4.22}%
\end{equation}

Moreover, in the case of weighted graphs $\left(  G,c\right)  $, the identity
\begin{equation}
\left(  \Delta u\right)  \left(  x\right)  =\sum_{y\sim x}c\left(  xy\right)
\left(  u\left(  x\right)  -u\left(  y\right)  \right)  \label{Eq4.23}%
\end{equation}
holds.
\end{theorem}

\begin{proof}
[Proof of \emph{(}\ref{Eq4.21}\emph{)}]%
\begin{align*}
\left(  \Delta v_{x}\right)  \left(  y\right)   &  =_{\left(  \text{by
}\left(  \ref{Eq4.19}\right)  \right)  }\left\langle \delta_{y},v_{x}%
\right\rangle \\
&  =_{\left(  \text{by }\left(  \ref{Eq4.6}\right)  \right)  }\delta
_{y}\left(  x\right)  -\delta_{y}\left(  o\right) \\
&  =\left(  \delta_{x}-\delta_{0}\right)  \left(  y\right)  \text{, }%
\end{align*}
which is (\ref{Eq4.21}).
\end{proof}

\begin{proof}
[Proof of \emph{(}\ref{Eq4.22}\emph{)}]Consider $u=\sum_{x}\xi_{x}v_{x}$, a
finite linear combination, $\xi_{x}\in\mathbb{C}$; then
\begin{align*}
\left\langle u,\Delta u\right\rangle  &  =\underset{x,y}{\sum\sum}\bar{\xi
}_{x}\xi_{y}\left\langle v_{x},\delta_{y}-\delta_{0}\right\rangle \\
&  =_{\left(  \text{by }\left(  \ref{Eq4.21}\right)  \right)  }\underset
{x,y}{\sum\sum}\bar{\xi}_{x}\xi_{y}\left\langle v_{x},\delta_{y}-\delta
_{0}\right\rangle \\
&  =_{\left(  \text{by }(\ref{Eq4.6})\right)  }\underset{x,y}{\sum\sum}%
\bar{\xi}_{x}\xi_{y}\left(  \delta_{x,y}+1\right) \\
&  =\sum_{x}\left\vert \xi_{x}\right\vert ^{2}+\left\vert \sum_{x}\xi
_{x}\right\vert ^{2}\geq0\text{.}%
\end{align*}

\end{proof}

\begin{proof}
[Proof of \emph{(}\ref{Eq4.23}\emph{)}]In the case of weighted graphs $\left(
G,c\right)  $
\begin{align*}
\left(  \Delta u\right)  \left(  x\right)   &  =_{\left(  \text{by }\left(
\ref{Eq4.19}\right)  \right)  }\left\langle \delta_{x},u\right\rangle
_{\mathcal{H}_{E}}\\
&  =_{\left(  \text{by }\left(  \ref{Eq4.12}\right)  \right)  }\frac{1}%
{2}\underset{%
\genfrac{}{}{0pt}{}{s,t}{s\sim t}%
}{\sum\sum}c\left(  s,t\right)  \left(  \delta_{x}\left(  s\right)
-\delta_{x}\left(  t\right)  \right)  \left(  u\left(  s\right)  -u\left(
t\right)  \right) \\
&  =\sum_{s\sim x}c\left(  sx\right)  \left(  u\left(  x\right)  -u\left(
s\right)  \right)
\end{align*}
which is the desired formula (\ref{Eq4.23}).
\end{proof}

\begin{corollary}
\label{Cor4.8}Let $\left(  \mathcal{H},X,o\right)  $ be as in the theorem, and
let $\Delta$ be the operator in \emph{(}\ref{Eq4.19}\emph{)}. Let
$\Delta^{\operatorname*{clo}}$ be the graph-closure of $\Delta$.

Then the domain of $\Delta^{\operatorname*{clo}}$ is contained in $\ell
^{2}(X)\cap\ell^{1}(X)$ where $\ell^{2}\cap\ell^{1}$ is understood with regard
to counting measure on $X$.

Note that
\begin{equation}
\left\{  \delta_{x}\right\}  \subseteq\mathcal{H} \label{Eq4.24}%
\end{equation}
is part of the assumption in the corollary.
\end{corollary}

\begin{proof}
\textbf{Step 1}. We saw that if
\[
u=\sum_{x\in G^{0}\diagdown\left(  0\right)  }\xi_{x}v_{x}%
\]
is a finite summation with $\xi_{x}\in\mathbb{C}$, then $\xi_{x}=\left(
\Delta u\right)  \left(  x\right)  $. Hence by the theorem,
\begin{equation}
\left\langle u,\Delta u\right\rangle _{E}=\sum_{x\in X^{\ast}}\left\vert
\left(  \Delta u\right)  \left(  x\right)  \right\vert ^{2}+\left\vert
\underset{x\in X^{\ast}}{\sum\Delta u\left(  x\right)  }\right\vert
^{2}\text{.} \label{Eq4.25}%
\end{equation}

\textbf{Step 2.} A simple approximation argument shows that (\ref{Eq4.25})
extends to be valid also for all $u\in\operatorname*{dom}\left(
\Delta^{\operatorname*{clo}}\right)  $.

To see this, note that by (\ref{Eq4.6}), the $\left\Vert \cdot\right\Vert
_{E}$-norm convergence implies pointwise convergence. If a sequence $\left(
u_{n}\right)  \subset\mathcal{D}_{V}$ is chosen such that
\[
\left\Vert u-u_{n}\right\Vert _{E}\rightarrow0;~\left\Vert \Delta u_{n}%
-\Delta^{\operatorname*{clo}}u\right\Vert _{E}\rightarrow0,
\]
then we have pointwise of the corresponding functions on $X$, and we may apply
Fatou to the summations $\sum_{x}\left\vert \left(  \Delta u_{n}\right)
\left(  x\right)  \right\vert ^{2}$, and $\sum_{x}\left(  \Delta u_{n}\right)
\left(  x\right)  $.
\end{proof}

\begin{theorem}
\label{Theo4.9}Let $\left(  \mathcal{H},X\right)  $ be a reproducing kernel
Hilbert space with base point $o$, and assume \emph{(}\ref{Eq4.8}\emph{)}
holds. Set
\begin{align}
\mathcal{C}\text{{}}  &  \text{\emph{:}}=\mathcal{H}\ominus\left\{  \delta
_{x}|x\in X\right\} \label{Eq4.26}\\
&  \left(  =\left\{  u\in\mathcal{H}|\left\langle u,\delta_{x}\right\rangle
=0,~\forall x\right\}  \right)  ;\nonumber
\end{align}
then $\left(  \Delta,\mathcal{C}\right)  $ is a duality pair.
\end{theorem}

\begin{proof}
The claim is that $\mathcal{C}\subseteq\ker\left(  \Delta^{\ast}\right)  $;
see (\ref{Eq4.1}). But by (\ref{Eq4.21}), $\Delta$ maps its domain
$\mathcal{D}_{V}$ into $\mathcal{C}^{\bot}$, so if $h\in\mathcal{C}$, then
\[
\left\langle \Delta u,h\right\rangle =0\text{ for }\forall u\in\mathcal{D}%
_{V}\text{.}%
\]
Hence $h\in\operatorname*{dom}\left(  \Delta^{\ast}\right)  $, and
\[
\left\langle u,\Delta^{\ast}h\right\rangle =0,~\forall u\in\mathcal{D}%
_{V}\text{.}%
\]
But $\mathcal{D}_{V}$ is dense in $\mathcal{H}$ by (\ref{Eq4.7}), and we
conclude that $\Delta^{\ast}h=0$, i.e., that $h\in\ker\left(  \Delta^{\ast
}\right)  $.
\end{proof}

\noindent\textbf{Remark 3.4(a) revisited.} Even though in Example
(\ref{Eq4.11}), $\delta_{x}$ is not in $\mathcal{H}$, the operator $\Delta
$:$=-\left(  \frac{d}{dx}\right)  ^{2}$ on the domain $\mathcal{D}%
:=C_{c}^{\infty}\left(  0,1\right)  $ is still hermitian and (\ref{Eq4.21})
holds. However, this candidate for domain $\mathcal{D}$ is not dense in
$\mathcal{H}$; in fact the function $f\left(  x\right)  =x$ is in
$\mathcal{H}\ominus C_{c}^{\infty}\left(  0,1\right)  $.

\section{The essential selfadjointness problem for a pair of hermitian
operators in duality\label{EssentProb}}

Let $\left(  \mathcal{H},X\right)  $ be a reproducing kernel Hilbert space
with base point $o,$ and assume that
\begin{equation}
\delta_{x}\in\mathcal{H}\text{ for all }x\in X\text{.} \label{Eq5.1}%
\end{equation}
Let $\Delta$ be the associated hermitian operator.
\begin{equation}
\left(  \Delta f\right)  \left(  x\right)  \text{:}=\left\langle \delta
_{x},f\right\rangle ,~x\in X\text{.} \label{Eq5.2}%
\end{equation}
Let $v_{x}$:$=k\left(  \cdot,x\right)  $ be the functions in $\mathcal{H}$
derived from the reproducing kernel $k\left(  \cdot,\cdot\right)  $ for
$\mathcal{H}$. Set
\begin{align*}
\mathcal{F}  &  \text{:}=\text{closed }\operatorname*{span}\left\{  \delta
_{x}|x\in X\right\}  ;\\
\mathcal{C}  &  \text{:}=\mathcal{H}\ominus\mathcal{F}\text{;}\\
\mathcal{D}_{F}  &  \text{:}=\operatorname*{span}\left\{  \delta_{x}|x\in
X\right\}  ;\text{ and }\\
\mathcal{D}_{V}  &  \text{:}=\operatorname*{span}\left\{  v_{x}|x\in X^{\ast
}\right\}  ;
\end{align*}
where $X^{\ast}$:$=X\diagdown\left(  0\right)  $, and where \textquotedblleft
span\textquotedblright\ means \textquotedblleft finite linear
combinations.\textquotedblright

It follows from Theorem \ref{Theo3.6} and Theorem \ref{Theo4.7} that the
prescriptions
\begin{equation}
\Delta_{F}\text{:}=\Delta|_{\mathcal{D}_{F}}\text{, } \label{Eq5.3}%
\end{equation}
meaning restriction; and
\begin{equation}
\Delta_{V}\text{:}=\Delta|_{\mathcal{D}_{V}} \label{Eq5.4}%
\end{equation}
yield hermitian operators with dense domain; the domain $\mathcal{D}_{F}$ of
$\Delta_{F}$ is dense in $\mathcal{H}$.

Let $\Delta_{F}^{\operatorname*{clo}}$ be the closure of $\Delta_{F}$ with
domain $\operatorname*{dom}\left(  \Delta_{F}^{\operatorname*{clo}}\right)  $;
and similarly $\Delta_{V}^{\operatorname*{clo}}$ for the closure of
$\Delta_{V}$ as a densely defined hermitian operator in $\mathcal{H}$.

Finally, set
\begin{equation}
\mathcal{D}_{H}\text{:}\text{=}\mathcal{D}_{V}+\mathcal{C}\text{, }
\label{Eq5.5}%
\end{equation}
and
\begin{equation}
\Delta_{H}\left(  u+h\right)  \text{:}=\Delta_{V}u\text{ for all }%
u\in\mathcal{D}_{V}\text{ and }h\in\mathcal{C}\text{.} \label{Eq5.6}%
\end{equation}

We proved in Theorem \ref{Theo3.6} that
\begin{equation}
\operatorname*{dom}\left(  \Delta_{H}^{\ast}\right)  =\left\{  \psi
\in\operatorname*{dom}\left(  \Delta_{V}^{\ast}\right)  |\Delta_{V}^{\ast}%
\psi\in\mathcal{F}\right\}  \text{.} \label{Eq5.7}%
\end{equation}

\begin{definition}
\label{Def5.1}A hermitian operator $\Delta$ with dense domain $\mathcal{D}$ is
a Hilbert space $\mathcal{H}$ is said to be \emph{selfadjoint} iff
$\Delta=\Delta^{\ast};$ and it is said to be \emph{essentially selfadjoint}
iff $\Delta^{\operatorname*{clo}}$ is selfadjoint, where $\Delta
^{\operatorname*{clo}}$ means the \emph{(}graph\emph{)} closure of $\Delta$.
\end{definition}

By a theorem of von Neumann (\cite{vN32}, \cite{ReSi75}) $\Delta$ is
essentially selfadjoint iff there are values $\lambda_{\pm}\in\mathbb{C}$,
$\operatorname*{Im}\lambda_{+}>0$, $\operatorname*{Im}\lambda_{-}<0$ such that
the two equations
\begin{equation}
\Delta^{\ast}\psi_{\pm}=\lambda_{\pm}\psi_{\pm} \label{Eq5.8}%
\end{equation}
in $\mathcal{H}$ only have the zero-solutions $\psi_{\pm}=0$. The solutions
$\psi_{\pm}$ to (\ref{Eq5.8}) form the \textit{deficiency spaces}, and their
respective dimensions are called the \textit{deficiency indices}.

If $\Delta$ is further semibounded, i.e., $\left\langle u,\Delta
u\right\rangle \geq0$ for all $u\in\mathcal{D}$, then for essential
selfadjointness it is enough to verify that the equation
\begin{equation}
\Delta^{\ast}\psi=-\psi\label{Eq5.9}%
\end{equation}
has only the zero-solution $\psi=0$ in $\mathcal{H}$. (It is understood in
(\ref{Eq5.8}) and (\ref{Eq5.9}) that the vectors $\psi_{\pm}$ and $\psi$ are
assumed to be in $\operatorname*{dom}\left(  \Delta^{\ast}\right)  $.

\begin{theorem}
\label{Theo5.2}Consider the two operators $\Delta_{F}$ in $\mathcal{F}$, and
$\Delta_{H}$ in $\mathcal{H}$ above, equation, \emph{(}\ref{Eq5.3}\emph{)} and
\emph{(}\ref{Eq5.6}\emph{)}, respectively.

Fix $\lambda\in\mathbb{C}$, with $\operatorname*{Im}\lambda\not =0$, or if
$\Delta_{V}$ is semibounded, $\operatorname*{Re}\lambda<0;$ then a function
$\psi\in\mathcal{H}$ satisfies
\begin{equation}
\Delta_{H}^{\ast}\psi=\lambda\psi\label{Eq5.10}%
\end{equation}
if and only if $\psi\in\mathcal{F}$, and
\begin{equation}
\Delta_{F}^{\ast}\psi=\lambda\psi\text{.} \label{Eq5.11}%
\end{equation}

\end{theorem}

\begin{proof}
The reasoning is on (\ref{Eq5.7}) and the previous considerations. Indeed we
have the following two-way implications:

\begin{itemize}
\item $\psi$ satisfies (\ref{Eq5.10}).\medskip

\item[ ] $\Updownarrow\medskip$

\item $\psi\in\operatorname*{dom}\left(  \Delta_{V}^{\ast}\right)  ,$
$\Delta_{V}^{\ast}\psi\in\mathcal{F}$, and
\begin{equation}
\Delta_{V}^{\ast}\psi=\lambda\psi\text{.} \label{Eq5.12}%
\end{equation}

\vspace{-10pt}

\item[ ] $\Updownarrow$\medskip

\item $\psi\in\mathcal{F}$, and $\Delta_{F}^{\ast}\psi=\lambda\psi$.
\end{itemize}

In the last step we used that $\lambda\not =0$, so a solution $\psi$ to
(\ref{Eq5.12}) with $\Delta_{V}^{\ast}\psi\in\mathcal{F}$ must be in
$\mathcal{F}$.
\end{proof}

\begin{corollary}
\label{Cor5.3}Let the two operators $\Delta_{F}$ in $\mathcal{F}$, and
$\Delta_{V}$ in $\mathcal{H}$ be as above\emph{;} and let $\Delta_{H}$ be the
extension of $\Delta_{V}$ from \emph{(}\ref{Eq5.6}\emph{)}.

Then the following two properties are equivalent\emph{:}

\begin{enumerate}
\item[(i)] $\Delta_{H}$ is essentially selfadjoint\emph{;} and

\item[(ii)] $\Delta_{F}$ is essentially selfadjoint.
\end{enumerate}
\end{corollary}

\begin{proof}
The equivalence (i)$\Leftrightarrow$(ii) is immediate from the theorem; given
von Neumann's theory of deficiency spaces; see (\ref{Eq5.10}) and
(\ref{Eq5.11}) above.
\end{proof}

\begin{remark}
\label{Rem5.4}In many applications \emph{(}see \cite{JoPe08}\emph{)} it's
easier to verify essential selfadjointness for $\Delta_{F}$ than it is for
$\Delta_{H}$.

This is an instance of our duality theory\emph{:} A comparison of
\emph{restrictions} and \emph{extensions}.
\end{remark}

\begin{corollary}
\label{Cor5.5}Consider the two operators $\Delta_{F}$ in $\mathcal{F}$ and
$\Delta_{V}$ in $\mathcal{H}$. Let $\Delta_{H}$ be the extension of
$\Delta_{V}$ defined in \emph{(}\ref{Eq5.6}\emph{)}. We assume that
$\mathcal{C}\subseteq\ker\left(  \Delta_{V}^{\ast}\right)  $. Then the
following four affirmations are equivalent\emph{:}

\begin{enumerate}
\item[(i)] $\Delta_{V}^{\ast}$ maps its domain into $\mathcal{F}%
=\mathcal{H}\ominus\mathcal{C}$.

\item[(ii)] $\Delta_{V}^{\ast}=\Delta_{H}^{\ast}$.

\item[(iii)] $\Delta_{V}^{\operatorname*{clo}}=\Delta_{H}^{\operatorname*{clo}%
}$.

\item[(iv)] $\mathcal{C}\subseteq\operatorname*{dom}\left(  \Delta
_{V}^{\operatorname*{clo}}\right)  $.
\end{enumerate}
\end{corollary}

\begin{proof}
(i)$\Rightarrow$(ii). In general $\Delta_{V}\subseteq\Delta_{H}$ so
$\Delta_{H}^{\ast}\subseteq\Delta_{V}^{\ast}$, and
\begin{equation}
\operatorname*{dom}\left(  \Delta_{H}^{\ast}\right)  =\left\{  \psi
\in\operatorname*{dom}\left(  \Delta_{V}^{\ast}\right)  |\Delta_{V}^{\ast}%
\psi\in\mathcal{F}\right\}  , \label{Eq5.13}%
\end{equation}
so if (i) holds, then $\operatorname*{dom}\left(  \Delta_{H}^{\ast}\right)
=\operatorname*{dom}\left(  \Delta_{V}^{\ast}\right)  $, and (ii) follows.

(ii)$\Rightarrow$(iii). Take adjoints in (ii), and we get
\[
\Delta_{V}^{\operatorname*{clo}}=\Delta_{V}^{\ast\ast}=\Delta_{H}^{\ast\ast
}=\Delta_{H}^{\operatorname*{clo}}%
\]
which is condition (iii).

(iii)$\Rightarrow$(iv). A simple limit consideration applied to (\ref{Eq5.6})
yields the following:
\begin{equation}
\operatorname*{dom}\left(  \Delta_{H}^{\operatorname*{clo}}\right)
=\operatorname*{dom}\left(  \Delta_{V}^{\operatorname*{clo}}\right)
+\mathcal{C} \label{Eq5.14}%
\end{equation}
which proves (iii)$\Leftrightarrow$(iv). Since $\Delta_{H}^{\ast}$ maps into
$\mathcal{F}$ by (\ref{Eq5.13}), it follows that (iv)$\Rightarrow$(i).
\end{proof}

There are many applications of selfadjoint extension operators; a major reason
being that the Spectral Theorem applies to each selfadjoint extension, while
it does not apply to a hermitian non-selfadjoint operator.

The operator $\Delta$ we consider here is semibounded on its dense domain, so
it has semibounded selfadjoint extensions with the same bound, for example the
Friedrichs extension $\Delta_{Fr}$; see \cite{DuSc88}.

The following applies to any one of the semibounded selfadjoint extension
$\Delta_{S}$ of $\Delta$. Given $\Delta_{S}$, there is a projection valued
measure $E_{S}\left(  \cdot\right)  $ defined on the Borel-sets $\mathcal{B}$
in $[0,\infty)$ and mapping into projections in $\mathcal{H}$; i.e., each
$P$:$=E_{S}\left(  B\right)  ,~B\in\mathcal{B}$ satisfies $P=P^{\ast}=P^{2}$;
and we have
\begin{gather}
\Delta_{S}=\int\limits_{0}^{\infty}\lambda E_{S}\left(  d\lambda\right)
,~I_{\mathcal{H}}=\int\limits_{0}^{\infty}E_{S}\left(  d\lambda\right)
\text{, and }\label{Eq5.15}\\
\int\limits_{0}^{\infty}\left\Vert E_{S}\left(  d\lambda\right)  u\right\Vert
_{\mathcal{H}}^{2}=\left\Vert u\right\Vert _{\mathcal{H}}^{2}\text{ for all
}u\in\mathcal{H}. \label{Eq5.16}%
\end{gather}

\begin{definition}
\label{Def5.6}Let $\left(  \mathcal{H},X,o\right)  $ be a relative reproducing
kernel Hilbert space satisfying the conditions above, and let $\Delta$,
$\Delta_{S}$ be associated operators with the listed properties. Let
$\mathcal{H}_{\mathbb{R}}$ be a real form of $\mathcal{H}$, and set
\begin{equation}
\mathcal{S}\text{\emph{:}}=\left\{  u\in\mathcal{H}_{\mathbb{R}}%
|\int\limits_{0}^{\infty}\lambda^{2p}\left\Vert E_{S}\left(  d\lambda\right)
u\right\Vert ^{2}<\infty,~\text{for all }p\in\mathbb{N}\right\}  \text{.}
\label{Eq5.17}%
\end{equation}
Recall
\begin{equation}
u\in\operatorname*{dom}\left(  \Delta_{S}^{p}\right)  \Leftrightarrow
\left\Vert \Delta_{S}^{p}u\right\Vert ^{2}=\int\limits_{0}^{\infty}%
\lambda^{2p}\left\Vert E_{S}\left(  d\lambda\right)  u\right\Vert ^{2}%
<\infty\text{.} \label{Eq5.18}%
\end{equation}

\end{definition}

We turn $\mathcal{S}$ into a Fr\'{e}chet space with the seminorms
\begin{equation}
\left\Vert u\right\Vert _{p}\text{:}=\left\Vert \Delta_{S}^{p}u\right\Vert
_{\mathcal{H}},~\text{for }u\in\mathcal{S},~\text{and }p\in\mathbb{N}\text{.}
\label{Eq5.19}%
\end{equation}
and we denote the dual of $\left(  \mathcal{S},~\left\Vert \cdot\right\Vert
_{p}\right)  _{p\in\mathbb{N}}$ by $\mathcal{S}^{\prime}$ for tempered distributions.

As a result we get the following \textit{Gelfand triple} \cite{JoPe08}:
\begin{equation}
\mathcal{S}\subseteq\mathcal{H}_{\mathbb{R}}\subseteq\mathcal{S}^{\prime}
\label{Eq5.20}%
\end{equation}
with the two inclusions in (\ref{Eq5.20}) representing continuous embeddings.

The cylinder sets $\left(  \subseteq\mathcal{S}^{\prime}\right)  $ generate a
sigma-algebra $\mathcal{B}$:$=\mathcal{B}\left(  \mathcal{S}^{\prime}\right)
$, and there is an associated (Wiener-) measure $W$ defined on $\mathcal{B}$
determined uniquely by the following identity:
\begin{equation}
\int\limits_{\mathcal{S}^{\prime}}e^{i\left\langle u,\xi\right\rangle
}dW\left(  \xi\right)  =e^{-\frac{1}{2}\left\Vert u\right\Vert _{\mathcal{H}%
}^{2}},~\text{for all }u\in\mathcal{H}_{\mathbb{R}}\text{.} \label{Eq5.21}%
\end{equation}

In the exponent on the LHS in (\ref{Eq5.21}), the expression $\left\langle
u,\cdot\right\rangle $ will be denoted as a function $\tilde{u}$ on
$\mathcal{S}^{\prime}$. We have
\begin{equation}
\left\langle u_{1},u_{2}\right\rangle _{\mathcal{H}_{\mathbb{R}}}%
=\int\limits_{\mathcal{S}^{\prime}}\tilde{u}_{1}\tilde{u}_{2}~dW\text{,}
\label{5.22}%
\end{equation}
and
\begin{equation}
\int\limits_{\mathcal{S}^{\prime}}\tilde{u}~dW=0\text{,} \label{Eq5.22}%
\end{equation}
for all $u_{1},u_{2},u\in\mathcal{H}_{\mathbb{R}}$.

If $\mu$ is a signed measure on $\mathcal{S}^{\prime}$, we denote its Fourier
transform
\begin{equation}
\hat{\mu}\left(  u\right)  \text{:}=\int\limits_{\mathcal{S}^{\prime}%
}e^{i\left\langle u,\xi\right\rangle }~d\mu\left(  \xi\right)  \text{;}
\label{Eq5.23}%
\end{equation}
or simply $\int_{\mathcal{S}^{\prime}}e^{i\tilde{u}\left(  \cdot\right)
}~d\mu\left(  \cdot\right)  $.

\begin{definition}
\label{Def5.7}We shall need the Hermite polynomials $\left(  H_{n}\right)
_{n\in\mathbb{N}_{0}}$ given by $H_{0}\equiv1$, $H_{1}\left(  x\right)  =1-x$,
and
\begin{equation}
\frac{d}{dx}e^{-\frac{x^{2}}{2}}=H_{n}\left(  x\right)  e^{-\frac{x^{2}}{2}};
\label{Eq5.24}%
\end{equation}
so
\[
H_{n+1}\left(  x\right)  =\frac{d}{dx}H_{n}\left(  x\right)  -xH_{n}\left(
x\right)  \text{.}%
\]

\end{definition}

\begin{lemma}
\label{Lem5.8}Let $\mathcal{H}$, $\mathcal{S}$, $\mathcal{S}^{\prime}$ be as
in \emph{(}\ref{Eq5.20}\emph{)}, let $f\in\mathcal{H}_{\mathbb{R}}$ be
given\emph{;} and set
\begin{equation}
d\mu\left(  \cdot\right)  \text{\emph{:}}=\tilde{f}\left(  \cdot\right)
dW\left(  \cdot\right)  \text{.} \label{Eq5.25}%
\end{equation}
Then for the Fourier transform, we have
\begin{equation}
\widehat{d\mu}\left(  u\right)  =i\tilde{f}\left(  u\right)  e^{-\frac
{\left\Vert u\right\Vert ^{2}}{2}}=i\left\langle f,u\right\rangle
_{\mathcal{H}}e^{-\frac{\left\Vert u\right\Vert _{\mathcal{H}}^{2}}{2}}\text{
for all }u\in\mathcal{S}\text{.} \label{Eq5.26}%
\end{equation}

\end{lemma}

\begin{proof}
Let $f\in\mathcal{H}_{\mathbb{R}},$ $u\in\mathcal{S}$, and $\varepsilon
\in\mathbb{R}_{+}$. By (\ref{Eq5.21}), we then have
\begin{equation}
\int\limits_{\mathcal{S}^{\prime}}e^{i\left\langle u+\varepsilon
f,\cdot\right\rangle }dW\left(  \cdot\right)  =e^{-\frac{1}{2}\left\Vert
u+\varepsilon f\right\Vert _{\mathcal{H}}^{2}}\text{.} \label{Eq5.27}%
\end{equation}
An application of $\left.  \frac{d}{d\varepsilon}\right\vert _{\varepsilon=0}$
to both sides in (\ref{Eq5.27}) then yields
\[
\int\limits_{\mathcal{S}^{\prime}}i\tilde{f}\left(  \cdot\right)
e^{i\left\langle u,\cdot\right\rangle }dW\left(  \cdot\right)  =-\left\langle
u,f\right\rangle e^{-\frac{\left\Vert u\right\Vert ^{2}}{2}}\text{.}%
\]

By virtue of (\ref{Eq5.23}) and (\ref{Eq5.24}), this formula is equivalent to
(\ref{Eq5.26}); i.e., the conclusion in the lemma.
\end{proof}

\begin{proposition}
\label{Prop5.9}Let $\left(  \mathcal{H},X,o\right)  $ be a relative
reproducing kernel Hilbert space. For $x,y\in X$, $x\not =y$, let $w_{x,y}%
\in\mathcal{H}$ be the solution to
\begin{equation}
\left\langle w_{x,y},u\right\rangle =u\left(  x\right)  -u\left(  y\right)
,~\forall u\in\mathcal{H}\text{.} \label{Eq5.28}%
\end{equation}

Then
\begin{equation}
\widehat{\left(  w_{x,y}dW\right)  }\left(  u\right)  =i\left(  u\left(
x\right)  -u\left(  y\right)  \right)  e^{-\frac{\left\Vert u\right\Vert ^{2}%
}{2}},~\forall u\in\mathcal{H}_{\mathbb{R}}\text{.} \label{Eq5.29}%
\end{equation}

\end{proposition}

\begin{proof}
We have
\[
\widehat{\left(  w_{x,y}dW\right)  }\left(  u\right)  =_{\left(  \text{by
}\left(  \ref{Eq5.26}\right)  \right)  }i\left\langle w_{x,y},u\right\rangle
e^{-\frac{\left\Vert u\right\Vert ^{2}}{2}}=_{\left(  \text{by }\left(
\ref{Eq5.28}\right)  \right)  }i\left(  u\left(  x\right)  -u\left(  y\right)
\right)  e^{-\frac{\left\Vert u\right\Vert ^{2}}{2}}%
\]
which is the desired formula (\ref{Eq5.29}).
\end{proof}

\begin{theorem}
\label{Theo5.10}Let $\mathcal{H}$, $f$, $u$, and $W$ be as described above.
Then
\begin{equation}
\widehat{\tilde{f}\left(  \cdot\right)  ^{n}dW}\left(  \cdot\right)  \left(
u\right)  =H_{n}\left(  f,u\right)  e^{-\frac{1}{2}\left\Vert u\right\Vert
^{2}} \label{Eq5.30}%
\end{equation}
where $H_{n},n\in\mathbb{N}_{0}$, are the Hermite polynomials in
\emph{(}\ref{Eq5.24}\emph{);}
\begin{align*}
H_{1}\left(  f,u\right)   &  \emph{:}=i\left\langle f,u\right\rangle ,~\\
H_{2}\left(  f,u\right)   &  \emph{:}=\left\Vert f\right\Vert ^{2}%
-\left\langle u,f\right\rangle ^{2},
\end{align*}
etc.
\end{theorem}

\begin{proof}
The reader may check that the theorem follows from Lemma \ref{Lem5.8},
combined with the recursive Hermite formulas (\ref{Eq5.24}) in Definition
\ref{Def5.7}. Indeed we must apply $\left.  \left(  \frac{d}{d\varepsilon
}\right)  ^{n}\right\vert _{\varepsilon=0}$ recursively to the RHS in
(\ref{Eq5.27}):
\[
\left.  \frac{d}{d\varepsilon}\right\vert _{\varepsilon=0}e^{-\frac{1}%
{2}\left\Vert u+\varepsilon f\right\Vert ^{2}}=-\left\langle u,f\right\rangle
e^{-\frac{\left\Vert u\right\Vert ^{2}}{2}},
\]
and
\[
\left.  \left(  \frac{d}{d\varepsilon}\right)  ^{2}\right\vert _{\varepsilon
=0}e^{-\frac{1}{2}\left\Vert u+\varepsilon f\right\Vert ^{2}}=\left(
\left\langle u,f\right\rangle ^{2}-\left\Vert f\right\Vert ^{2}\right)
e^{-\frac{1}{2}\left\Vert u\right\Vert ^{2}}\text{.}%
\]

\end{proof}

\begin{definition}
\label{Def5.11}Let $\left(  \mathcal{H},X,o\right)  $ be a relative
reproducing kernel Hilbert space satisfying condition \emph{(}\ref{Eq5.1}%
\emph{)}. \emph{(}It follows then that $X$ is discrete!\emph{)} Let
$\mathcal{S}^{\prime}$ be the \emph{real} space in the Gelfand triple
\emph{(}\ref{Eq5.20}\emph{)}, and let $W$ be the corresponding Wiener measure
determined by \emph{(}\ref{Eq5.21}\emph{)}. By a \emph{boundary point} for
$\left(  \mathcal{H},X,o\right)  $ we mean a measure $\beta$ on $\mathcal{S}%
^{\prime}$ such that there is a sequence $x_{1},x_{2},\cdots$ in $X$
satisfying
\begin{equation}
\lim\limits_{n\rightarrow\infty}\left(  u\left(  x_{n}\right)  -u\left(
o\right)  \right)  =\int\limits_{\mathcal{S}^{\prime}}\tilde{u}~d\beta\text{
for all }u\in\mathcal{H}. \label{Eq5.31}%
\end{equation}

\end{definition}

The following result is different from the classical Rieman-Lebesgue theorem,
but it is inspired by it. First, for $x\in X$ set $v_{x}$:$=w_{x,o}=$ the
dipole for the pair of points $x,o$ in $X$.

\begin{corollary}
\label{Cor5.12}Let $\left(  \mathcal{H},X,o\right)  $ and $\left\{
v_{x}\right\}  _{x\in X^{\ast}}$ be a reproducing kernel system subject to the
conditions listed above\emph{;} and let $\beta$ be a boundary point.

Then there is a sequence $x_{1},x_{2},\cdots$ in $X$ such that
\begin{equation}
\lim\limits_{n\rightarrow\infty}\left(  v_{x_{n}}dW\right)  \widehat
{\;}\left(  u\right)  =i\int\limits_{\mathcal{S}^{\prime}}\tilde{u}%
~d\beta\cdot e^{-\frac{\left\Vert u\right\Vert ^{2}}{2}}\text{.}
\label{Eq5.32}%
\end{equation}

\end{corollary}

\begin{proof}
Let $\beta$ be a boundary point as indicated. Pick $\left(  x_{n}\right)
_{n\in\mathbb{N}}\subset X$ such that (\ref{Eq5.31}) holds. When this is
substituted into (\ref{Eq5.29}) from Proposition \ref{Prop5.9}, the desired
conclusion (\ref{Eq5.32}) follows.
\end{proof}

\section{Computing deficiency-spaces for a pair of Hermitian operators in
duality\label{ComputingSpaces}}

The setting will be as in the two previous sections. We are given a
reproducing kernel Hilbert space $\left(  \mathcal{H},X,o\right)  $ with
base-point, and we introduce the three associated hermitian operators
$\Delta_{F}$ in $\mathcal{F}$; and $\Delta_{V}$ (with $\Delta_{V}%
\subseteq\Delta_{H}$) and $\Delta_{H}$ densely defined operators in
$\mathcal{H}$.

We saw in Theorem \ref{Theo5.2} that it is frequently easier to compute
deficiency spaces for $\Delta_{F}$ than it is for the other two operators in
the larger ambient Hilbert space $\mathcal{H}$. But in all cases $\mathcal{H}$
may be somewhat intractable because it is determined by a fixed reproducing
kernel $k\left(  \cdot,\cdot\right)  $, and the spanning functions
$v_{x}\left(  \cdot\right)  $:$=k\left(  \cdot,x\right)  $ are far from
forming an orthogonal system in $\mathcal{H}$; in fact in my examples, turning
$\left\{  v_{x}|x\in X^{\ast}\right\}  $ into a \textit{frame} still leaves
with poor frame-bound estimates.

As before, here we set $X^{\ast}$:$=X\diagdown\left\{  o\right\}  $.

We will now examine the fundamental property,
\begin{equation}
\delta_{x}\in\mathcal{H}\text{, }\forall x\in X\text{.} \label{Eq6.1}%
\end{equation}

Our aim is to represent $\delta_{x}$ as an expansion in $\left\{  v_{y}|y\in
X^{\ast}\right\}  $.

To avoid difficulties with \textquotedblleft bad\textquotedblright%
\ frame-bounds, we add the following restricting assumption,

For all $x$, we have
\begin{equation}
\#\left\{  y\in X|\left\langle \delta_{x},\delta_{y}\right\rangle
\not =0\right\}  <\infty\text{.} \label{Eq6.2}%
\end{equation}
We will see in the next section that (\ref{Eq6.2}) corresponds to a
finite-degree restriction on a graph build from the system $\left(
\mathcal{H},X,o\right)  $.

\begin{proposition}
\label{Prop6.1}Let $\left(  \mathcal{H},X,o\right)  $ and $\left\{  v_{x}|x\in
X^{\ast}\right\}  $ be as specified above\emph{;} see also section
\ref{EssentProb} for additional details.

Then
\begin{equation}
\delta_{x}=\left\Vert \delta_{x}\right\Vert _{\mathcal{H}}^{2}v_{x}+\sum_{y\in
X\diagdown\left\{  o,x\right\}  }\left\langle \delta_{x},\delta_{y}%
\right\rangle v_{y}\text{.} \label{Eq6.3}%
\end{equation}

\end{proposition}

\begin{proof}
Note that with assumption (\ref{Eq6.2}), we have ruled out infinite summations
occurring on the R.H.S. in (\ref{Eq6.3}). However, it is possible to relax
condition (\ref{Eq6.2}), and this will be taken up in a subsequent paper.
\end{proof}

We will need the following:

\begin{lemma}
\label{Lem6.2}If some $u\in\mathcal{H}$ has a finite representation
\begin{equation}
u=\sum_{x\in X^{\ast}}\xi_{x}v_{x}, \label{Eq6.4}%
\end{equation}
finite summation, and with $\xi_{x}\in\mathbb{C}$\emph{;} then
\begin{equation}
\xi_{x}=\left(  \Delta u\right)  \left(  x\right)  \left(  \text{:}%
=\left\langle \delta_{x},u\right\rangle _{\mathcal{H}}\right)  \text{.}
\label{Eq6.5}%
\end{equation}

\end{lemma}

\begin{proof}
Let $y\in X^{\ast}$, and compute
\begin{align*}
\left\langle \delta_{y},u\right\rangle _{H}  &  =_{(\text{by \ref{Eq6.4}}%
)}\sum_{x\in X^{\ast}}\xi_{x}\left\langle \delta_{y},v_{x}\right\rangle
_{\mathcal{H}}\\
&  =\sum_{x\in X^{\ast}}\xi_{x}\left(  \delta_{y}\left(  x\right)  -\delta
_{y}\left(  0\right)  \right) \\
&  =\xi_{y}\text{;}%
\end{align*}
and therefore
\[
\xi_{y}=\left\langle \delta_{y},u\right\rangle _{\mathcal{H}}=\left(  \Delta
u\right)  \left(  y\right)
\]
which is the desired conclusion.
\end{proof}

\bigskip

\begin{proof}
[Proof of Proposition 6.1 resumed]With the lemma, we now compute the L.H.S. in
(\ref{Eq6.3}):
\begin{align*}
\delta_{x}  &  =_{\left(  \text{by the lemma}\right)  }\sum_{y}\left(
\Delta\delta_{x}\right)  \left(  y\right)  v_{y}\\
&  =\left(  \Delta\delta_{x}\right)  \left(  x\right)  v_{x}+\sum_{y\not =%
x}\left(  \Delta\delta_{x}\right)  \left(  y\right)  v_{y}\\
&  =_{\left(  \text{by}\left(  \ref{Eq6.5}\right)  \right)  }\left\langle
\delta_{x},\delta_{x}\right\rangle _{\mathcal{H}}v_{x}+\sum_{y\not =%
x}\left\langle \delta_{y},\delta_{x}\right\rangle _{\mathcal{H}}v_{y}%
\end{align*}
which is the desired formula (\ref{Eq6.3}).
\end{proof}

\begin{definition}
\label{Def6.3}Let $\left(  \mathcal{H},X,o\right)  $ be as above, and set
\begin{equation}
c\left(  x\right)  =\max\left(  \left\Vert \delta_{x}\right\Vert
_{\mathcal{H}}^{2},\sum_{y\not =x}\left\vert \left\langle \delta_{x}%
,\delta_{y}\right\rangle _{\mathcal{H}}\right\vert \right)  \text{.}
\label{Eq6.6}%
\end{equation}

Let $\ell^{2}\left(  X,c\right)  $ be the $\ell^{2}$-space with $c\left(
\cdot\right)  $ as weight, i.e., all $\xi$\emph{:}\thinspace$X\rightarrow
\mathbb{C}$ such that
\begin{equation}
\sum_{x}c\left(  x\right)  \left\vert \xi_{x}\right\vert ^{2}=\text{\emph{:}%
}\left\Vert \xi\right\Vert _{\ell_{\left(  c\right)  }^{2}}^{2}<\infty\text{.}
\label{Eq6.7}%
\end{equation}

\end{definition}

\begin{theorem}
\label{Theo6.4}Let $\left(  \mathcal{H},X,o\right)  $ be as above, and let the
function $c\left(  \cdot\right)  $ be defined by \emph{(}\ref{Eq6.6}\emph{)}.
Then $\ell^{2}\left(  c\right)  $ is contractively embedded in $\mathcal{H}$.
\end{theorem}

\begin{proof}
Since $\mathcal{H}$ is a Hilbert space, we shall state the embedding of
$\ell^{2}\left(  c\right)  $ into $\mathcal{H}$ instead as a mapping into
$\mathcal{H}^{\ast}=($the dual of $\mathcal{H})\simeq\mathcal{H}$.

For $\xi\in\ell^{2}\left(  c\right)  $, set
\begin{equation}
L\left(  \xi\right)  \text{:}\,u\longmapsto\sum_{x}\xi_{x}\cdot\left(  \Delta
u\right)  \left(  x\right)  \text{.} \label{Eq6.8}%
\end{equation}

We will show that the summation on the R.H.S. in (\ref{Eq6.8}) is absolutely
convergent and that
\begin{equation}
\sum_{x}\left\vert \xi_{x}\cdot\left(  \Delta u\right)  \left(  x\right)
\right\vert ^{2}\leq\left\Vert \xi\right\Vert _{\ell_{\left(  c\right)  }^{2}%
}^{2}\cdot\left\Vert u\right\Vert _{\mathcal{H}}^{2}\text{.} \label{Eq6.9}%
\end{equation}
The conclusion in the theorem follows from this, and an application of Riesz'
lemma to $\mathcal{H}$.

By the theorem, we have
\begin{equation}
\delta_{x}=\sum_{y}\left\langle \delta_{x},\delta_{y}\right\rangle
_{\mathcal{H}}w_{x,y} \label{Eq6.10}%
\end{equation}
where $w_{x,y}$:$=v_{x}-v_{y}$; and
\begin{align*}
&  \sum_{x}\left\vert \xi_{x}\left(  \Delta u\right)  \left(  x\right)
\right\vert \\
&  =\sum_{x}\left\vert \xi_{x}\left\langle \delta_{x},u\right\rangle
_{\mathcal{H}}\right\vert \\
&  =_{\left(  \text{by }\left(  \ref{Eq6.10}\right)  \right)  }\sum_{\times
}\sum_{y}\left\vert \xi_{x}\left\langle \delta_{x},\delta_{y}\right\rangle
_{\mathcal{H}}\left\langle w_{x,y},u\right\rangle _{\mathcal{H}}\right\vert \\
&  =_{\left(  \text{Fubini}\right)  }\sum_{y}\sum_{x}\left\vert \xi
_{x}\left\langle \delta_{x}\delta_{y}\right\rangle _{\mathcal{H}}\left\langle
w_{x,y},u\right\rangle _{\mathcal{H}}\right\vert \\
&  =\sum_{y}\sum_{x}\left\vert \xi_{x}\right\vert \sqrt{\left\vert
\left\langle \delta_{x},\delta_{y}\right\rangle _{\mathcal{H}}\right\vert
}\sqrt{\left\vert \left\langle \delta_{x},\delta_{y}\right\rangle
_{\mathcal{H}}\right\vert }\left\vert u\left(  x\right)  -u\left(  y\right)
\right\vert \\
&  \leq_{\left(  \text{Schwarz}\right)  }\sum_{y}\left(  \sum_{x}\left\vert
\xi_{x}\right\vert ^{2}\cdot\left\vert \left\langle \delta_{x},\delta
_{y}\right\rangle \right\vert \right)  ^{\frac{1}{2}}\cdot\left(  \sum
_{x}\left\vert \left\langle \delta_{x},\delta_{y}\right\rangle _{\mathcal{H}%
}\right\vert \cdot\left\vert u\left(  x\right)  -v\left(  y\right)
\right\vert ^{2}\right)  ^{\frac{1}{2}}\\
&  \leq_{\left(  \text{Schwarz}\right)  }\left(  \sum_{x}\left\vert \xi
_{x}\right\vert ^{2}c\left(  x\right)  \cdot\sum_{y}\left\vert \left\langle
\delta_{x},\delta_{y}\right\rangle _{\mathcal{H}}\right\vert \cdot\left\vert
u\left(  x\right)  -v\left(  y\right)  \right\vert ^{2}\right)  ^{\frac{1}{2}%
}\\
&  \leq\left\Vert \xi\right\Vert _{\ell_{\left(  c\right)  }^{2}}%
\cdot\left\Vert u\right\Vert _{\mathcal{H}}\text{,}%
\end{align*}
which is the desired estimate (\ref{Eq6.9}).
\end{proof}

Let $\left(  \mathcal{H},X,o\right)  $ be a relative reproducing kernel
Hilbert space such that (\ref{Eq6.1}) is satisfied; and let $\Delta$ be the
associated operator from (\ref{Eq5.2}). Since vectors in $\mathcal{H}$ are
determined from differences (via dipoles, see equation (\ref{Eq5.28}))
intuitively one would expect the constant function $1$ to be represented by
zero in $\mathcal{H}$.

The next result offers an operator theoretic answer to this question.

\begin{definition}
\label{Def6.5}A family of finite subsets $\left(  F_{k}\right)  _{k\in
\mathbb{N}}$ is said to be an \emph{exhaustion} or a \emph{filtration} in $X$
if
\begin{equation}
F_{1}\subset F_{2}\subset\cdots F_{k}\subset F_{k+1}\subset\cdots
\label{Eq6.11}%
\end{equation}
and
\begin{equation}
\bigcup\limits_{k=1}^{\infty}F_{k}=X\text{.} \label{Eq6.12}%
\end{equation}
Let
\begin{equation}
f_{k}\text{\emph{:}}=\chi_{F_{k}}=\sum_{x\in F_{k}}\delta_{x}\text{.}
\label{Eq6.13}%
\end{equation}

Now define the boundaries $\operatorname*{bd}F_{k}$ for each $k$ as
follows\emph{:}
\begin{equation}
\operatorname*{bd}F_{k}=\left\{  x\in F_{k}|\exists y\in F_{k}^{c}\text{ with
}\left\langle \delta_{x},\delta_{y}\right\rangle \not =0.\right\}
\label{Eq6.14}%
\end{equation}
For simplicity we will assume finite degrees, i.e., assume that \emph{(}%
\ref{Eq4.17}\emph{)} is satisfied for all $x\in X$. By $F_{k}^{c}$ we mean the
complement, i.e., $F_{k}^{c}$\emph{:}$=X\diagdown F_{k}$.

For functions $\psi$ on $X$, define a \emph{normal derivative }$\frac
{\partial\psi}{\partial n}$ referring to the filtration\emph{:}
\begin{equation}
\frac{\partial\psi}{\partial n}\left(  x\right)  =\sum_{y\sim x}\left\langle
\delta_{x},\delta_{y}\right\rangle \left(  \psi\left(  x\right)  -\psi\left(
y\right)  \right)  \label{Eq6.15}%
\end{equation}
where $y\sim x$ means $y\not =x$ and $\left\langle \delta_{x},\delta
_{y}\right\rangle \not =0$. Moreover for $x\in F_{k}$, set
\begin{equation}
\left(  \frac{\partial\psi}{\partial n}\right)  _{k}\left(  x\right)
=\sum_{y\in F_{k}^{c}}\left\langle \delta_{x},\delta_{y}\right\rangle \left(
\psi\left(  x\right)  -\psi\left(  y\right)  \right)  \text{.} \label{Eq6.16}%
\end{equation}

\end{definition}

\begin{lemma}
\label{Lem6.6}Let $\left\langle \cdot,\cdot\right\rangle $ be the inner
product in $\mathcal{H}$, and let $F\subset X$ be a finite subset. Then for
$\psi\in\mathcal{H}$, we have the identity
\begin{equation}
\left\langle \chi_{F},\psi\right\rangle =\sum_{x\in F}\left(  \frac
{\partial\psi}{\partial n}\right)  _{F}\left(  x\right)  \text{.}
\label{Eq6.17}%
\end{equation}

\end{lemma}

\begin{proof}%
\begin{align*}
\left\langle \chi_{F},\psi\right\rangle  &  =\underset{x\not =y}{\sum\sum
}\left\langle \delta_{x},\delta y\right\rangle \left(  \chi_{F}\left(
x\right)  -\chi_{F}\left(  y\right)  \right)  \left(  \psi\left(  x\right)
-\psi\left(  y\right)  \right) \\
&  =\sum_{x\in F}\sum_{y\in F^{c}}\left\langle \delta_{x},\delta
_{y}\right\rangle \left(  \psi\left(  x\right)  -\psi\left(  y\right)  \right)
\\
&  =\sum_{x\in F}\left(  \frac{\partial\psi}{\partial n}\right)  _{F}\left(
x\right)  \text{.}%
\end{align*}

\end{proof}

\begin{theorem}
\label{Theo6.7}Let $\left(  \mathcal{H},X,o\right)  $ be as above, and let
$\left(  F_{k}\right)  _{k\in\mathbb{N}}$ be a filtration. Then
\[
f_{k}\text{\emph{:}}=\chi_{F_{k}}\in\mathcal{H}%
\]
converges to zero weakly if and only if
\begin{equation}
\lim\limits_{k\rightarrow\infty}\sum_{x\in F_{k}}\left(  \frac{\partial\psi
}{\partial n}\right)  _{F_{k}}\left(  x\right)  =0\text{ for all }\psi
\in\mathcal{H}. \label{Eq6.18}%
\end{equation}

\end{theorem}

\begin{proof}
The conclusion (\ref{Eq6.18}) is immediate from the lemma.
\end{proof}

\begin{corollary}
\label{Cor6.8}Let $\mathcal{H}=\mathcal{H}_{E}$ where $\mathcal{H}_{F}$ is the
energy Hilbert space coming from a weighted graph $\left(  G,c\right)  $ with
$G^{0}=$ the set of vertices, and $G^{1}=$ the set of edges, i.e.,
\[
\left\langle u,v\right\rangle _{\mathcal{H}_{E}}=\underset{x\sim y}{\sum\sum
}c\left(  xy\right)  \left(  \overline{u\left(  x\right)  }-\overline{u\left(
y\right)  }\right)  \left(  v\left(  x\right)  -v\left(  y\right)  \right)
\]
and
\[
\left(  \Delta u\right)  \left(  x\right)  =\sum_{y\sim x}c\left(  xy\right)
\left(  u\left(  x\right)  -u\left(  y\right)  \right)  \text{.}%
\]
Then for every filtration $\left(  F_{k}\right)  $ in $G^{0}$, $\chi_{F_{k}%
}\rightarrow0$ as $k\rightarrow\infty$, with weak convergence in
$\mathcal{H}_{E}$.
\end{corollary}

\begin{proof}
By Theorem \ref{Theo6.7}, we only need to prove that the limit property
(\ref{Eq6.18}) is satisfied; but this follows in turn from the proof of
Theorem 4.7; specifically the proof of (\ref{Eq4.21}) in this theorem.
\end{proof}

\begin{example}
\label{Ex6.9}Let $G$ be the graph $\mathbb{Z}^{d}$ with nearest
neighbors\emph{;} i.e., $x\sim y$ for pairs of points $x$ and $y$ in
$x=\left(  x_{1},\cdots,x_{d}\right)  $, $y=\left(  y_{1},\cdots,y_{d}\right)
$ and the two only differ on one coordinate place, i.e., $\exists i$ such that
$\left\vert x_{i}-y_{i}\right\vert =1$. For $x\sim y$ set $c\left(  xy\right)
=1$.
\end{example}

For filtration, let
\[
F_{k}\text{:}=\left[  -k,k\right]  ^{d}\cap\mathbb{Z}^{d}\text{,}%
\]
and $f_{k}$:$=\chi_{F_{k}}$. Then
\[
\left\Vert f_{k}\right\Vert _{\mathcal{H}_{E}}^{2}=\left(  2d\right)
\cdot\left(  2k\right)  ^{d-1}.
\]
In particular, it follows that $\left(  f_{k}\right)  _{k\in\mathbb{N}}$ is
not a Cauchy sequence in $\mathcal{H}_{E}$.

\section{Concluding remarks and applications\label{ConcludingRem}}

We saw that every weighted graph $\left(  G,c\right)  $ with finite degrees
gives rise to a reproducing kernel Hilbert space $\left(  \mathcal{H}%
,X\right)  $ with $X=G^{c}\diagdown\left(  o\right)  $. Here $G^{c}$ denotes
the set of vertices in $G$, and $o$ is a chosen (and fixed) base-point for
$G^{0}$. To see this we introduce the graph Laplacian
\begin{equation}
\left(  \Delta u\right)  \left(  x\right)  \text{:}=\sum_{y\sim x}c\left(
xy\right)  \left(  u\left(  x\right)  -u\left(  y\right)  \right)
\label{Eq7.1}%
\end{equation}
with
\begin{equation}
c\left(  x\right)  \text{:}=\sum_{y\sim x}c\left(  xy\right)  . \label{Eq7.2}%
\end{equation}
Equation (\ref{Eq7.1}) then takes the form
\begin{equation}
\left(  \Delta u\right)  \left(  x\right)  =c\left(  x\right)  u\left(
x\right)  -\sum_{y\sim x}c\left(  xy\right)  u\left(  y\right)  \label{Eq7.3}%
\end{equation}
for all functions $u$ on $G^{0}$.

In section \ref{Extensions}, and in \cite{JoPe08}, we introduced the
associated energy Hilbert space $\mathcal{H}_{E}$ with its inner product
$\left\langle \cdot,\cdot\right\rangle _{E}$ and norm $\left\Vert
\cdot\right\Vert _{E}$. We showed that for every $x$, there is a unique
$v_{x}\in\mathcal{H}_{E}$ such that
\begin{equation}
\left\langle v_{x},f\right\rangle =f\left(  x\right)  -f\left(  0\right)
,~\forall f\in\mathcal{H}_{E}. \label{Eq7.4}%
\end{equation}
Setting $w_{x,y}$:$=v_{x}-v_{y}$, we get
\begin{equation}
\left\langle w_{x,y},f\right\rangle =f\left(  x\right)  -f\left(  y\right)
\text{.} \label{Eq7.5}%
\end{equation}
Furthermore, the Dirac functions $\delta_{x}$ satisfy
\begin{equation}
\delta_{x}\in\mathcal{H}_{E}\text{, and }\left\Vert \delta_{x}\right\Vert
_{E}^{2}=c\left(  x\right)  \text{ for all }x\in G^{0}\text{.} \label{Eq7.6}%
\end{equation}

\begin{theorem}
\label{Theo7.1}\emph{(}a\emph{)} Let $\left(  \mathcal{H},X\right)  $ be a
reproducing kernel Hilbert space of functions on a set $X$. Let $o\in X$ be a
base-point. Let $k\left(  \cdot,\cdot\right)  $ be the reproducing kernel for
$\left(  \mathcal{H},X,o\right)  $, and set
\begin{equation}
v_{x}\left(  y\right)  \text{\emph{:}}=k\left(  y,x\right)  \text{ for }x\in
X^{\ast}\text{.} \label{Eq7.7}%
\end{equation}
Then $\left(  v_{x}\right)  _{x\in X^{\ast}}$ satisfies
\begin{equation}
\left\langle v_{x},f\right\rangle =f\left(  x\right)  -f\left(  0\right)
,~\forall f\in\mathcal{H},~\forall x\in X^{\ast}\text{.} \label{Eq7.8}%
\end{equation}

\emph{(}b\emph{)} The following two affirmations are equivalent\emph{:}

\begin{enumerate}
\item[(i)] $\left(  \mathcal{H},X,o\right)  $ satisfies\emph{:}

\begin{itemize}
\item $\delta_{x}\in\mathcal{H}$, $\forall x\in X$.

\item For every $x\in X$, we have
\begin{equation}
\#\left\{  y\in X|\left\langle \delta_{x},\delta_{y}\right\rangle
\not =0\right\}  <\infty\text{.} \label{Eq7.9}%
\end{equation}

\item The following identity holds\emph{:}
\begin{equation}
\left\Vert \delta_{x}\right\Vert ^{2}=-\sum_{y\in X}\left\langle \delta
_{x},\delta_{y}\right\rangle \text{.} \label{Eq7.10}%
\end{equation}

\end{itemize}

\item[(ii)] There is a weighted graph $\left(  G,c\right)  $ with finite
degrees such that $X=G^{0}$\emph{;}
\begin{equation}
G^{1}=\left\{  \left(  x,y\right)  |\left\langle \delta_{x},\delta
_{y}\right\rangle \not =0\right\}  ; \label{Eq7.11}%
\end{equation}
and
\begin{equation}
c\left(  xy\right)  \emph{:}=-\left\langle \delta_{x},\delta_{y}\right\rangle
,~\forall\left(  xy\right)  \in G^{1}. \label{Eq7.12}%
\end{equation}

\end{enumerate}

\emph{(}c\emph{)} If the conditions in \emph{(}i\emph{)} or \emph{(}ii\emph{)}
are satisfied, the Laplace operator
\[
\left(  \Delta u\right)  \left(  x\right)  \emph{:}=\left\langle \delta
_{x},u\right\rangle
\]
satisfies \emph{(}\ref{Eq7.1}\emph{)}.
\end{theorem}

\begin{proof}
(a) This is already proved in section \ref{Extensions}.

(b) (i)$\Rightarrow$(ii). Assume that some reproducing kernel Hilbert space
$\left(  \mathcal{H},X,o\right)  $ with base-point $o$ satisfies the three
conditions (bullet points) listed in (i). We will construct a weighted graph
$\left(  G,c\right)  $ with $G^{0}$:$=X$; $G^{1}$ we take to be the set in
(\ref{Eq7.11}); and we set
\[
c\left(  xy\right)  \text{:}=-\left\langle \delta_{x},\delta_{y}\right\rangle
\]
as in (\ref{Eq7.12}); and $c\left(  x\right)  $:$=\left\Vert \delta
_{x}\right\Vert ^{2}$. We then have the implication: (10)$\Rightarrow$(12). As
a result, the axioms for weighted graphs are satisfied for this particular
$\left(  G,c\right)  $, and the degrees are finite by assumption (\ref{Eq7.9}).

We set
\begin{equation}
\left(  \Delta u\right)  \left(  x\right)  \text{:}=\left\langle \delta
_{x},u\right\rangle ,~u\in\mathcal{H} \label{Eq7.13}%
\end{equation}
which is possible by the first assumption in (i).

It remains to prove that then (\ref{Eq7.1}) is satisfied.
\end{proof}

\begin{lemma}
\label{Lem7.2}Let $x,y\in X$, and suppose $x\not =y$. Let
\[
u\in\operatorname*{span}\left\{  \delta_{z}|z\in X\right\}  ,~u\text{\emph{:}%
}=\sum_{z}\xi_{z}\delta_{z}%
\]
\emph{(}a finite linear combination $\xi\in\mathbb{C}$\emph{)}.

Then
\begin{equation}
\xi_{x}-\xi_{y}=\left\langle w_{x,y},u\right\rangle \label{Eq7.14}%
\end{equation}
with $w_{x,y}$ from \emph{(}\ref{Eq7.5}\emph{)}.
\end{lemma}

\begin{proof}
We have
\begin{align*}
\left\langle w_{x,y},u\right\rangle  &  =\sum_{z}\xi_{z}\left\langle
w_{x,y},\delta_{z}\right\rangle \\
&  =_{\left(  \text{by }\left(  \ref{Eq7.5}\right)  \right)  }\sum_{z}\xi
_{z}\left(  \delta_{z}\left(  x\right)  -\delta_{z}\left(  y\right)  \right)
\\
&  =\xi_{x}-\xi_{y}\text{.}%
\end{align*}

\end{proof}

\begin{proof}
[Proof of Theorem \ref{Theo7.1} resumed]~\smallskip\newline\textit{Case 1.} If
$u\in\mathcal{H}$, and
\begin{equation}
u\bot\left\{  \delta_{z}|z\in X\right\}  \text{, } \label{Eq7.15}%
\end{equation}
then set $\Delta u=0$.\smallskip

\noindent\textit{Case 2.} If $u=\sum_{z}\xi_{z}\delta_{z}$ is a finite linear
combination as in Lemma \ref{Lem7.2}, we may compute $\left(  \Delta u\right)
\left(  x\right)  $ from the assumptions as follows:
\begin{align*}
\left(  \Delta u\right)  \left(  x\right)   &  =\sum_{z}\xi_{z}\left(
\Delta\delta_{z}\right)  \left(  x\right) \\
&  =\xi_{x}\left(  \Delta\delta_{x}\right)  \left(  x\right)  +\sum_{z\not =%
x}\xi_{z}\left(  \Delta\delta_{z}\right)  \left(  x\right) \\
&  =_{\left(  \text{by }\left(  \ref{Eq7.13}\right)  \right)  }\xi
_{x}\left\langle \delta_{x},\delta_{x}\right\rangle +\sum_{z\not =x}\xi
_{z}\left\langle \delta_{x},\delta_{z}\right\rangle \\
&  =_{\left(  \text{by }\left(  \ref{Eq7.10}\right)  \right)  }-\xi_{x}%
\sum_{z\not =x}\left\langle \delta_{x},\delta_{z}\right\rangle +\sum
_{z\not =x}\xi_{z}\left\langle \delta_{x},\delta_{z}\right\rangle \\
&  =\sum_{z\not =x}\left(  \xi_{z}-\xi_{x}\right)  \left\langle \delta
_{x},\delta_{z}\right\rangle \\
&  =_{\left(  \text{by }\left(  \ref{Eq7.14}\right)  \right)  }-\sum
_{z\not =x}\left\langle w_{x,z},u\right\rangle \left\langle \delta_{x}%
,\delta_{z}\right\rangle \\
&  =_{\left(  \text{by }\left(  \ref{Eq7.5}\right)  \right)  }-\sum_{z\not =%
x}\left(  u\left(  x\right)  -u\left(  z\right)  \right)  \left\langle
\delta_{x},\delta_{z}\right\rangle \\
&  =_{\left(  \text{by }\left(  \ref{Eq7.12}\right)  \right)  }\sum_{z\not =%
x}c\left(  xz\right)  \left(  u\left(  x\right)  -u\left(  z\right)  \right)
\text{,}%
\end{align*}
which is the desired formula (\ref{Eq7.1}).

The proof of the converse formula (ii)$\Rightarrow$(i) amounts to showing that
every weighted graph $\left(  G,c\right)  $ with finite degrees yields a
reproducing kernel Hilbert space representation as stated. But with $\left(
G,c\right)  $ given, we may take $\mathcal{H}$:$=\mathcal{H}_{E}$, as in
section \ref{Extensions}; $X$:$=G^{0}=$ the set of vertices. A direct
computation then yields the formulas
\[
\left\langle \delta_{x},\delta_{y}\right\rangle _{E}=\left\{
\begin{array}
[c]{ll}%
\;c\left(  x\right)  & \text{if }y=x\\
-c\left(  xy\right)  & \text{if }y\sim x\\
\;0 & \text{for other cases, i.e., }y\not =x\text{ and }y\not \sim x\text{.}%
\end{array}
\right.
\]
and, as a result, we get $\left(  \mathcal{H}_{E},G^{0},o\right)  $ as a
reproducing kernel Hilbert space with base point $0$, and reproducing kernel
\[
k\left\langle x,y\right\rangle =\left\langle v_{x},v_{y}\right\rangle
_{E}\text{.}%
\]
Finally, it follows that equation (\ref{Eq7.10}) will then be satisfied.
\end{proof}

\begin{center}
\textsc{Acknowledgment}
\end{center}

The author is happy to thank his colleagues (for suggestions) in the math
physics and operator theory seminars at the University of Iowa, where he
benefitted from numerous helpful discussions of various aspects of the results
in the paper, and their applications. And he thanks Doug Slauson for the typesetting.

This work supported in part by the US National Science Foundation.

\bibliographystyle{amsplain}
\bibliography{acompat,Jorgen}

\end{document}